\documentclass[draftcls, 12pt, a4paper, onecolumn]{IEEEtran}
\usepackage{amssymb, amsfonts, amsmath}
\usepackage{cite}
\usepackage{graphicx}
\usepackage{citesort}

\newcommand{\nocu}{L}

\newcommand{\fgain}{h}

\newcommand{\pfgain}{\gamma}

\newcommand{\errsig}{\sigma_e}
\newcommand{\de}{{d_e}}

\newcommand{\npfgain}{\omega}

\newcommand{\im}{i}
\newcommand{\ang}{\theta}
\newcommand{\ident}{\mathbf{I}}

\newcommand{\txsig}{\boldsymbol{x}}
\newcommand{\conste}{\mathcal{X}}
\newcommand{\consize}{M}
\newcommand{\block}{b}
\newcommand{\noblock}{B}
\newcommand{\power}{P}
\newcommand{\npower}{\pi}

\newcommand{\infoX}{I_{\conste}}
\newcommand{\rate}{R}
\newcommand{\mgain}{r}

\newcommand{\duni}{d^{\rm uni}}

\newcommand{\Pout}{\mathsf{P}_{\rm out}}
\newcommand{\real}{\mathbb{R}}
\newcommand{\outset}{\mathcal{O}}

\newcommand{\incvec}[2]{#1^{(#2)}}
\newcommand{\expectation}[2]{\mathbb{E}_{#1}\left[#2\right]}

\newcommand{\openone}[1]{\leavevmode\hbox{\small1\normalsize\kern-.33em1}\left\{#1\right\}}

\newcommand{\mtxsig}{\boldsymbol{X}}
\newcommand{\mrxsig}{\boldsymbol{Y}}
\newcommand{\mnoise}{\boldsymbol{W}}
\newcommand{\chmat}{\boldsymbol{H}}

\newcommand{\notx}{{N_\tau}}
\newcommand{\norx}{{N_\upsilon}}
\newcommand{\tx}{\tau}
\newcommand{\rx}{\upsilon}
\newcommand{\mindim}{n}
\newcommand{\maxdim}{m}
\newcommand{\eig}{\lambda}
\newcommand{\mpfgain}{\boldsymbol{\Gamma}}
\newcommand{\comp}{\mathbb{C}}

\newcommand{\mnpfgain}{\boldsymbol{\Omega}}
\newcommand{\diag}{{\rm diag}}
\newcommand{\goodset}{\mathcal{S}^{(\epsilon)}}

\newcommand{\tr}{{\rm tr}}
\newcommand{\del}{u}

\newcommand{\mat}[1]{{\boldsymbol #1}}
\newcommand{\EE}{\mathbb{E}}

\newcommand{\half}{\frac{1}{2}}

\newtheorem{theorem}{Theorem}
\newtheorem{cor}{Corollary}
\title{Causal/Predictive Imperfect Channel State Information in  Block-Fading Channels}
\author{Khoa D. Nguyen\thanks{K. D. Nguyen and
   N. Letzepis are with Institute for Telecommunications Research,
   University of South Australia, SPRI Building - Mawson Lakes Blvd.,
   Mawson Lakes SA 5095, Australia, e-mail: {\tt
     \{khoa.nguyen,nick.letzepis\}@unisa.edu.au}. Albert Guill\'en i F\`abregas is with
   the Department of Engineering, University of Cambridge, Trumpington street, Cambridge
   CB2 1PZ, UK, e-mail: {\tt
     guillen@ieee.org}. Lars K. Rasmussen is with the Communication Theory Lab, School of Electrical Engineering and the ACCESS Linneaus Center, Royal Institute of Technology, Stockholm, Sweden, e-mail:{\tt lars.rasmussen@ieee.org}.},~\IEEEmembership{Member,~IEEE}, Nick
 Letzepis,~\IEEEmembership{Member,~IEEE},\\ Albert Guill\'en i
 F\`abregas, ~\IEEEmembership{Senior~Member,~IEEE},  Lars K. Rasmussen, ~\IEEEmembership{Senior~Member,~IEEE}\thanks{This work was presented in part at International Symposium on Information Theory,  Austin, Texas, June 2010.}\thanks{ This work has been supported by the Sir Ross and Sir Keith Smith Fund, Cisco Systems; the Australian Research Council under ARC grants RN0459498, DP0881160; the Royal Society International Travel Grant 2009/R2, 2009/R4; the Swedish Research Council under VR grant 621-2009-4666 and European Community's Seventh Framework Programme (FP7/2007-2013) /ERC grant agreement No 228044.}}
\begin{document}
\IEEEoverridecommandlockouts 
\maketitle
%
%

\begin{abstract}
We consider a multi-input multi-output (MIMO) block-fading channel with a general model for channel state information at the transmitter (CSIT).  The model covers systems with causal CSIT, where only CSIT of past fading blocks is available, and predictive CSIT, where CSIT of some future fading blocks is available.  The optimal diversity-multiplexing tradeoff (DMT) and rate-diversity tradeoff (RDT) of the channel are studied under long-term power constraints.  The impact of imperfect (mismatched) CSIT on the optimal DMT and RDT is also investigated.   Our results show the outage diversity gain obtained by providing imperfect causal/predictive CSIT, leading to new insights into system design and analysis. 
\end{abstract}

%
%
\section{Introduction}
The mitigation of fading is a particularly challenging aspect in the
design of reliable and efficient wireless communication systems
\cite{Proakis1995}.  The methods available for dealing with fading-related impairments are influenced by many factors, where some of the most important are the time/frequency behaviour of the fading processes and system constraints in terms of delay and power.   For systems with
no delay constraints or systems subject to fast fading, the channel can be
considered \emph{ergodic}. In this case, long-interleaved fixed-rate
codes that do not exceed the channel capacity can be employed to ensure an
arbitrarily low probability of error
\cite{BiglieriProakisShamai1998,GoldSmithVaraiya1997}. In contrast, for slow fading channels with delay constraints, a transmitted codeword
may only experience a small finite number of independent fading
realisations and hence the channel is \textit{non-ergodic}. 

The \emph{block-fading} channel \cite{BiglieriProakisShamai1998,OzarowShamaiWyner1994} is a simple
model that captures the essence of non-ergodic channels. Here, each codeword
comprises a finite number of blocks, where each block experiences
an independent fading realisation, which remains constant within a given
block. In this case, the instantaneous input-output \emph{mutual
  information} is a random variable dependent on the underlying fading
distribution. For most fading statistics, the channel capacity is zero in the strict
Shannon sense as there is a non-zero \emph{outage probability} that a
fixed information rate is not supported \cite{BiglieriProakisShamai1998,OzarowShamaiWyner1994}. The outage
probability is the lowest achievable word error probability of codes
with sufficiently long block
length~\cite{MalkamakiLeib1999}. As such, a
rate-reliability tradeoff exists, whereby for a fixed number of
blocks, a high rate is penalised by a large error probability.

Most works that study the block-fading channel focus on adaptive transmission techniques in which the power and/or rate is adapted to the channel conditions subject to system constraints (see~\cite{NguyenThesis2009} for a recent review).  Adaptation, however, requires a certain degree of knowledge of the channel fades, also referred to as \emph{channel state information} (CSI), at the transmitter and receiver.  While it is a common assumption that CSI is available at the receiver, the availability of CSIT is system dependent.  Particularly, CSIT can be obtained through the reciprocal channel in time-division duplex systems \cite{KnoppCaire2002}, or via a dedicated feedback channel \cite{LoveHeathLauGesbertRaoAndrews2008}. 

 A large body of work considers full CSI at the transmitter (CSIT), i.e., the transmitter knows  values of the fades on all blocks.  Most notably, the works in \cite{CaireTariccoBiglieri1999,BiglieriCaireTaricco2001,NguyenGuillenRasmussen2010it} study systems with perfect CSIT, while systems with imperfect CSIT is analysed in \cite{KimSkoglund2007TC,KimNguyenGuillen2009}.   This approach has practical relevance for systems exhibiting a set of instantaneous parallel channels, such as Orthogonal Frequency Division Multiplexing (OFDM) systems.  The full CSIT assumption provides an upper bound to the performance of delay-limited communications.   However, there are practical situations where this assumption is invalid.  Specifically,  in time-varying fading channel, causality constraints impose that only CSIT of blocks up to the current block is available \cite{NegiCioffi2002}.  
Further delay in acquiring the CSIT may impose strictly causal CSIT.  In this case, CSIT is only available with a delay of a few fading blocks.  Strictly causal CSIT arises in systems that experience slow time-varying fading, where CSIT is obtained via a feedback, e.g. free-space optical systems \cite{LetzepisGuillen2008nov}.  Meanwhile, specific systems may allow the CSI of future fading blocks  to be  available at the transmitter.  For example, in mobile communications over slowly spatial fading channels, the CSI of the channel at a future location of the mobile device can be obtained and made available at the transmitter\footnote{The CSIT of future fading blocks can also be predicted from the current channels realisation when the fading blocks are statistically dependent.  However, correlated block-fading channels are not within the scope of this paper.}.  


In this paper, we analyse the outage performance of the multi-input multi-output (MIMO) block-fading channel with a general CSIT model, which includes systems with causal CSIT and systems where CSIT of future blocks is available.   We consider both systems with perfect CSIT and mismatched CSIT.  For systems with mismatched CSIT, the transmitter is provided with a  noisy version of channel fading gains, modelled with a Gaussian distribution as in \cite{KimNguyenGuillen2009}.  With perfect CSIT, power adaptation algorithms based on {\em dynamic programming} are proposed in \cite{NegiCioffi2002,ChenWong2009} for systems with causal CSIT.  Generalisations of algorithms can be derived for systems with perfect strictly causal CSIT, as well as systems with future-block CSIT.  However, dynamic programming does not provide much insight into the outage performance and may exhibit prohibitive complexity in many scenarios.  With imperfect CSIT a feasible adaptive power allocation rule is not known, even in the full CSIT case \cite{KimNguyenGuillen2009}.  We therefore study the asymptotic outage performance of the block-fading channel for various CSIT scenarios without explicitly solving the optimal power allocation problem.  In particular, we derive the optimal diversity-multiplexing tradeoff (DMT) and the optimal rate-diversity tradeoff (RDT) of the block-fading channel with long-term average power constraints.  From the tradeoffs we gain insights into the impact of causal and predictive CSIT, as well as of imperfect CSIT, on the asymptotic outage performance.   

 The analysis shows that reducing delays in obtaining the CSIT, or increasing the predictive CSIT to include additional future blocks, generally improves the DMT and RDT of the MIMO block-fading channel.  Similarly, improving the quality of the CSIT generally provides large gains in outage diversity at any multiplexing gain/transmission rate.  However,  at a given multiplexing gain/transmission rate, the optimal DMT/RDT may be dominated by either the number of CSIT blocks available or the CSIT quality.  Specifically, strictly causal CSIT provides gains in outage diversity only if the multiplexing gain/transmission rate is sufficiently small.  Furthermore, in agreement with the results in \cite{NguyenRasmussenGuillenLetzepis2009it}, the outage diversity of systems with strictly causal CSIT is always finite, even when the CSIT available is perfect. In contrast, systems with perfect predictive CSIT  achieve infinite outage diversity, which in many cases leads to a positive delay-limited capacity \cite{HanlyTse1998,NguyenGuillenRasmussen2010it}.  With imperfect predictive CSIT, increasing the number of predictive blocks improves the asymptotic outage performance, until the outage diversity is dominated by the CSIT noise.    These results highlight the roles of CSIT and its quality on the asymptotic outage performance of the MIMO block-fading channel, thus providing guidelines for system design. 

The remainder of the paper is organised as follows.  The system model is described in Section \ref{sec:system-model}, while Section \ref{sec:preliminaries} provides some preliminaries necessary for the paper.  In Section \ref{sec:divers-mult-trad}, the DMT of the block-fading channel with the optimal Gaussian input constellation is analysed, while Section \ref{sec:rate-divers-trad} studies the RDT of MIMO block-fading channels with arbitrary discrete input constellations.  Concluding remarks are given in Section \ref{sec:conclusion} and finally, proofs of various results are provided in the Appendices.  

The following notations are used in the paper. Scalar variables are denoted with lowercase symbols, lowercase and uppercase boldface symbols correspondingly denote vector and matrix variables. Sets are denoted with calligraphic symbols, while real and complex sets are denoted with $\real$ and $\comp$ correspondingly.  Hermitian and non-conjugate transposes are correspondingly denoted by $(\cdot)^\dag$ and $(\cdot)^T$.  The trace of a matrix is denoted by $\tr(\cdot)$; while $\diag(\mat A_1, \ldots, \mat A_n)$ denotes a block-diagonal matrix whose diagonal blocks are matrices $\mat A_1, \ldots, \mat A_n$.  Expectations are denoted by $\expectation{}{\cdot}$. The magnitude of $\xi$ is denoted as $|\xi|$, while $\lfloor\xi\rfloor (\lceil \xi\rceil)$ denotes the largest (smallest) integer smaller (greater) than $\xi$.  

%
%

\section{System Model}
\label{sec:system-model}
In this paper, we consider a MIMO block-fading channel,
with $\notx$ transmit and $\norx$ receive antennae.  For convenience, define $\mindim = \min\{\notx, \norx\}$ and $\maxdim= \max\{\notx, \norx\}$.  Binary data is encoded with a code of rate $R$ bits
per channel use, constructed over an alphabet $\mathcal{X} \subseteq \comp$. The resulting transmitted codeword consists of
$B$ blocks, where each block comprises of $L$ vector channel uses of size $\notx \times 1$. We denote $\mat{x}_b[l] \in \mathcal{X}^{\notx \times 1}$
as the $l$th transmitted symbol vector of block $b$, for $l =1,\ldots,L$ and $b = 1,\ldots, B$. The symbols are assumed to be drawn independently from $\mathcal{X}$ with unit average energy, i.e., $\EE\left[\mat{x}_b[l] \mat{x}^{\dagger}_b[l] \right] = \mat{I}_{\notx}$, where $\mat{I}_N$ the $N \times N$ identity matrix. 

We denote by $\mat{H}_b$ the $\norx \times \notx$ complex channel matrix for
block $b = 1,\ldots,B$. These matrices are drawn independently for each block, and remain fixed for the corresponding $L$  channel uses.  In addition, we assume the elements of the channel matrix
are independently, identically distributed (i.i.d.) complex Gaussian random variables (the Rayleigh fading channel model~\cite{Proakis1995})\footnote{The Rayleigh fading assumption is included mainly for notational simplicity.  The analysis can be extended to include a general fading distribution using the results in \cite{ZhaoMoMaWang2007}. }. Let $\widehat{\mat{H}}_\block$, available at the transmitter, be a noisy version of the true channel realisation $\mat{H}_\block$, so that 
\begin{equation}
\label{eq:mismatch_model}  \mat{H}_\block = \widehat{\mat{H}}_\block + \mat{E}_\block, \;\;\block=1, \ldots, \noblock, 
\end{equation}
where $\mat{E}_\block \in \comp^{\norx\times \notx}$ is the CSIT noise matrix (independent of $\widehat{\mat{H}}_\block$) whose entries are  i.i.d. complex Gaussian random variables with zero mean and variance $\errsig^2$.  This model has been motivated for imperfect CSIT in many practical communication scenarios \cite{VisotskyMadhow2001,JongrenSkoglundOttersten2002,LimLau2008,KimNguyenGuillen2009}.  As in \cite{LimLau2008,KimNguyenGuillen2009}, we assume that the CSIT noise variance decays as
\begin{equation}
  \errsig^2 = \power^{-\de} 
\end{equation}
for some $\de \geq 0$, where $\power$ is the power constraint to be defined subsequently.  For convenience, we introduce the normalised channel gains 
\begin{equation}
\label{eq:norm_CSIT_noise}
  \overline{\chmat}_\block= \frac{\sqrt{2}}{\errsig}\chmat_\block.  
\end{equation}
Given $\widehat{\chmat}_\block$, then $\overline\chmat_\block$ is a complex Gaussian matrix with mean $\frac{\sqrt{2}}{\errsig}\widehat{\mat{H}}_\block$ and entries having a scaled unit variance (unit variance on real and imaginary dimensions).  

 At the transmission of block $b$, we assume the transmitter only has knowledge of $\widehat{\mat{H}}^{(b-u)} = \diag(\widehat{\mat{H}}_1, \ldots, \widehat{\mat{H}}_{b-u})$, where $-B \leq u \leq B$ is an arbitrary fixed integer. When $0 < u \leq \noblock$, the parameter $u$ models the delay in obtaining CSI at the transmitter, due to, e.g., propagation and processing delays. When $-\noblock \leq u \leq 0$, the parameter $\del$ models the number of future blocks with predictive CSIT.  

We assume that the signal at each receive antenna is corrupted by independent, zero-mean unit-variance additive white Gaussian noise (AWGN). Hence, under these assumptions, the $b$th block of $\norx \times L$ received noisy
symbols is 
\begin{equation}
\label{eq:channel_model}
\mrxsig_\block= \chmat_\block \mat\power^{\half}_\block\left(\incvec{\widehat{\mat{H}}}{\block-\del}\right)  \mtxsig_\block+ \mnoise_\block, \;\;\;\;\;\; \block = 1, \ldots, \noblock, 
\end{equation}
where $\mtxsig_\block \in \conste^{\notx\times \nocu}$, $\mrxsig_\block \in \comp^{\norx\times
  \nocu}$; are correspondingly the transmit and receive signal in block $\block$; $\mnoise_\block \in \comp^{\norx\times \nocu}$ is the noise matrix whose elements are drawn i.i.d. from the zero-mean unit-variance Gaussian distribution; and $\mat\power_b\left(\widehat{\mat{H}}^{(b-u)}\right) \in \real_+^{\notx \times \notx}$ is a diagonal matrix whose $\tx^{\rm th}$ diagonal element denotes the power allocated to transmit antenna $\tx$ of block $b$.  We further assume that the receiver has perfect knowledge of $\mat{H}_\block$ and $\mat\power_\block\left(\incvec{\widehat\chmat}{\block-\del}\right)$ when receiving block $\block$. The power allocation is subject to the long-term power constraint,
\begin{equation}
 \expectation{}{\frac{1}{\noblock}\sum_{\block=1}^\noblock \tr\left(\mat\power_\block\left(\incvec{\widehat{\chmat}}{\block-\del}\right)\right)} \leq \power.  \label{eq:pow_constraint}
\end{equation}

\section{Preliminaries}
\label{sec:preliminaries}
The channel described by~\eqref{eq:channel_model} 
is not information stable under the assumption of quasi-static fading \cite{VerduHan1994}
and as a consequence, the capacity in the strict Shannon sense is zero. We
therefore study the information outage probability,
\begin{equation}
  \Pout(P, \rate) = \Pr\left\{\frac{1}{\noblock}\sum_{\block=1}^\noblock \infoX\left(\chmat_\block\mat\power^{\half}_\block\left(\incvec{\widehat\chmat}{\block-\del}\right)\right) < \rate\right\},  \label{eq:pout}
\end{equation}
which is a fundamental limit on the codeword error performance of any coding
scheme~\cite{BiglieriProakisShamai1998,OzarowShamaiWyner1994,MalkamakiLeib1999}. In~\eqref{eq:pout},
$\infoX(\boldsymbol{S})$ denotes the input-output mutual information
of a MIMO block-fading channel with input constellation $\conste$ and
channel matrix $\boldsymbol{S}$.  With the optimal Gaussian input constellation, 
\begin{equation}
\label{eq:info_MIMO_Gauss}
  \infoX(\boldsymbol{S}) = \log_2\det(\mat{I}_\norx + \mat{S}\mat{S}^\dag);
\end{equation}
while with a uniform discrete and fixed constellation $\conste$ of size $2^\consize$, 
\begin{equation}
\infoX(\boldsymbol{S})  = \consize\notx -
\label{eq:MIMO_info}\frac{1}{2^{\consize\notx}} \sum_{\boldsymbol x \in \conste^\notx}\expectation{\boldsymbol z}{\log_2\left(\sum_{\boldsymbol x' \in \conste^\notx}e^{-\|\boldsymbol{S}(\boldsymbol x - \boldsymbol x ') + \boldsymbol z\|^2 + \|\boldsymbol z\|^2} \right)}. 
\end{equation} 
Given mismatched CSIT  $\incvec{\widehat{\chmat}}{\block-\del}$, $\mat P_b\left(\incvec{\widehat\chmat}{b-\del}\right)$ is the solution to the minimisation problem
\begin{equation}
\begin{cases}
{\rm minimise} & \Pout(\power,\rate)  \\
{\rm subject~to:} &   \expectation{}{\frac{1}{\noblock}\sum_{\block=1}^\noblock \tr\left(\mat\power_\block\left(\incvec{\widehat\chmat}{\block-\del}\right)\right)} \leq \power.  \\ 
\end{cases} \label{eq:palloc_prob}
\end{equation}
For systems with full perfect CSIT, i.e., $\incvec{\chmat}{\noblock}$ is known at the transmitter prior to transmission, the optimal power allocation rule  and outage diversity is studied in \cite{BiglieriCaireTaricco2001,CaireTariccoBiglieri1999,NguyenGuillenRasmussen2010it}.  With perfect causal CSIT, i.e., $\incvec{\widehat\chmat}{\block-\del} = \incvec{\chmat}{\block-\del}$, and $u=0$, \eqref{eq:palloc_prob} can be solved via dynamic programming~\cite{NegiCioffi2002,ChenWong2009}. The extension to $u > 0$ or $-\noblock < \del < 0$ is also possible, although the problem becomes exceedingly difficult as $|u|$ increases.  With mismatched CSIT, the problem becomes even more challenging \cite{KimNguyenGuillen2009}.  However, as we shall see, it is possible to examine the asymptotic behaviour of $\Pout(\power,\rate)$ without explicitly solving~\eqref{eq:palloc_prob}.  In particular, for systems with a Gaussian input constellation, we study the DMT \cite{ZhengTse2003} $d(r)$ defined as
\begin{equation}
\label{eq:dmt_define}
  d(r)= \lim_{\power\to \infty} \frac{-\log \Pout(\power, \mgain\log_2\power)} {\log \power}, 
\end{equation}
where $ \mgain \in [0, \mindim]$ is the multiplexing gain.  Meanwhile, for systems with discrete input constellation $\conste$ of size $2^\consize$, we study the RDT $d(\rate)$, 
\begin{equation}
  \label{eq:diver_def} 
  d(\rate) \triangleq \lim_{\power \to \infty} \frac{-\log \Pout(\power, \rate)}{\log \power} 
\end{equation}
with $\rate \in (0, \consize\notx)$.  Note that the optimal RDT has been derived in \cite{KimNguyenGuillen2009} for the special case $\del=-\noblock$, where mismatched CSIT of all fading blocks is known prior to the transmission of each codeword.

For systems with uniform power allocation and Gaussian input constellation, it follows from \cite{ZhengTse2003} that the DMT $\duni(\mgain)$ is the piecewise linear curve connecting the points $(k, \noblock(\notx-k)(\norx-k)), k=0, \ldots, \mindim$.   Meanwhile, for systems with discrete input constellation $\conste$ of size $2^\consize$, the RDT is given by the Singleton bound \cite{NguyenThesis2009,LuKumar2005}
\begin{equation}
\duni(\rate)=  \norx d_S(\rate) \triangleq \norx\left(1+\left\lfloor \noblock\left(\notx-\frac{\rate}{\consize}\right)\right\rfloor\right).  \label{eq:singleton}
\end{equation}
Note that $\duni(R)$ is also the outage diversity of systems with short-term power constraint $\sum_{\block=1}^\noblock \tr\left(\mat\power_\block\left(\incvec{\widehat\chmat}{\block-\del}\right)\right) \leq \noblock\power$ \cite{NguyenGuillenRasmussen2010it}.    The optimal DMT and RDT of systems with long-term power constraints are investigated in the subsequent sections.

\section{Diversity-Multiplexing Tradeoff of Gaussian Input Channels}
\label{sec:divers-mult-trad}

\subsection{Causal CSIT}
 When causal CSIT is available, the achievable outage diversity of a MIMO block-fading channel with long-term power constraint is given in the following theorem.  
\begin{theorem}
  \label{the:imper-CSIT-MIMO-DMT} Consider transmission over the MIMO block-fading channel in \eqref{eq:channel_model} with multiplexing gain $\mgain \in [0, \mindim]$.   Assume that mismatched CSIT $\incvec{\widehat\chmat}{\block-\del}$ as modelled in \eqref{eq:mismatch_model}  is available at the transmission of block $\block$, where $\del > 0$ is the delay in obtaining CSIT.  Then the optimal DMT is lower bounded by 
  \begin{equation}
    d(\mgain, \de) \geq \min_{\mat k \in \left\{0, \ldots, \mindim\right\}^\noblock}   d_{\mat k}, 
  \end{equation}
where $d_{\mat k}$ is given as
\begin{equation}
  \label{eq:opt_dmt_mimo_causal}
  \begin{cases}
    {\rm Infimum~}& \noblock\maxdim\mindim\de+ \sum_{\block=1}^\noblock \sum_{i=1}^\mindim (2i-1+\maxdim-\mindim)\overline\npfgain_{\block, i} \\
    {\rm Subject~to~} 
    &\sum_{\block=1}^\noblock\sum_{i=1}^\mindim \left(\npower_{\block}\left(\incvec{\widehat{\mat\npfgain}}{\block-\del}\right)- \de- \overline\npfgain_{\block, i}\right)_+ < \noblock\mgain\\
    & \overline\npfgain_{\block, 1} \geq \ldots \geq \overline\npfgain_{\block, \mindim-k_\block} \geq 0 \geq \overline\npfgain_{\block, \mindim-k_\block+1} \geq \ldots\geq \overline\npfgain_{\block,\mindim} \geq - \de\\
    &\npower_{\block}\left(\incvec{\widehat{\mat\npfgain}}{\block-\del}\right) = 1+ \maxdim\mindim\de(\block-\del)_+ + \sum_{\block'=1}^{\block-\del} \sum_{i=\mindim-k_{\block'}+ 1}^\mindim (2i-1+\maxdim-\mindim) \overline\npfgain_{\block', i}
  \end{cases}
\end{equation}
\end{theorem}
\begin{proof}
  See appendix \ref{sec:causal-csit-proof}. 
\end{proof}
The optimisation problem in \eqref{eq:opt_dmt_mimo_causal} is linear, and can therefore easily be solved by a simplex algorithm\footnote{The problem is solved by first performing a minimisation, with inequalities constraints $<$ replaced by $\leq$, then the discontinuous points (if any) are taken care of based on the solution obtained.}.  The achievable DMT of a 2-by-2 MIMO block-fading channel with $\noblock = 4$ and causal CSIT with delay $\del=3$ is illustrated in Figure \ref{fig:DMT_MIMO_u3}.   The figure shows that the optimal DMT of a MIMO block-fading can be improved by increasing the quality of causal CSIT.  However, causal CSIT does not provide any gain in outage diversity for multiplexing gains $\mgain \geq 1.25$.  Outage diversity gains at high multiplexing gains $\mgain$ are only observed in systems with small $\del$, as illustrated in Figure \ref{fig:DMT_MIMO_del}.  Figure \ref{fig:DMT_MIMO_del} shows the achievable DMT of the same MIMO block-fading channel ($\notx=\norx=2, \noblock=4$) with perfect causal CSIT  for various delay $\del$, illustrating the significant impact of CSIT delay on the asymptotic outage performance.  Note in Figure \ref{fig:DMT_MIMO_u3} that for a given CSIT delay $\del$, a finite $\de$ is sufficient to achieve the optimal outage diversity of a corresponding system with perfect CSIT.  For example, with $\del=3$, $\de = 0.5$ exhibits the same DMT as systems with $\de = \infty$ for multiplexing gains $\mgain \geq 0.5$; also note that $\de =1$ is sufficient in terms of outage diversity  for all $\mgain \in [0, \mindim]$.  The $\de$ threshold required for obtaining the optimal DMT of a perfect causal CSIT will be analytically shown in the sequel for the vector channel case. 

Due to the complexity of the optimisation problem in \eqref{eq:opt_dmt_mimo_causal}, the impact of the various parameters such as $\de$ and $\del$ on the DMT curve is difficult to deduce in general.  We therefore consider the vector channel case ($\mindim=1$) as given in the following. 
\begin{cor}
\label{cor:DMT_SISO_causal_mismatch}
  Consider transmission with multiplexing gain $\mgain \in [0, 1]$ over the block-fading channel in \eqref{eq:channel_model} where $\mindim=1$.  Assume that causal mismatched CSIT $\incvec{\widehat\chmat}{\block-\del}$ as modelled in \eqref{eq:mismatch_model} is available at the transmission of block $\block$, where $\del$ is the delay in obtaining CSIT.  The optimal DMT is lower bounded by\footnote{In this case, it can be shown that the following achievable DMT is optimal following the same steps as in Theorem \ref{the:imper-CSIT-MIMO-DMT}.  Here the outage diversity is simpler to characterize and the converse is obtained since the  argument $\mat S \mat S^\dag$ of the mutual information expression in \eqref{eq:info_MIMO_Gauss} reduces to a scalar value.}
  \begin{equation}
    \label{eq:mismatch_DMT_causal_SISO}
    d(\mgain, \de)\geq
    \begin{cases}
      \maxdim\noblock(1-\mgain), & \noblock-\del-\noblock\mgain \leq 0\\
      \maxdim \sum_{i=1}^{\noblock-\lfloor\noblock\mgain\rfloor} a^\star_i,& {\rm otherwise}. 
   \end{cases}
 \end{equation}
 where $a^\star_i$'s are defined as follows, 
\begin{equation}
\label{eq:recursive_a}
  a^\star_i =
  \begin{cases}
1-\noblock\mgain+\lfloor\noblock\mgain\rfloor, &i =1\\     
1, &i=2, \ldots, \del\\
 a^\star_{i-1}+ \maxdim \min\left\{a^\star_{i-u}, \de \right\}
&i=u+1, \ldots, \noblock-\lfloor \noblock\mgain\rfloor
\end{cases}
\end{equation}
\end{cor}
\begin{proof}
See Appendix \ref{sec:causal-csit-vector}. 
\end{proof}
The expression of $d(\mgain, \de)$  in \eqref{eq:mismatch_DMT_causal_SISO} confirms the thresholds observed in Figure \ref{fig:DMT_MIMO_u3} for the vector channel case.  In particular,  no gain in outage diversity is obtained  by causal CSIT if the delay $\del$ satisfies $\del \geq \noblock(1-\mgain)$.  Furthermore, it also follows from \eqref{eq:mismatch_DMT_causal_SISO} and \eqref{eq:recursive_a} that increasing $\de$ beyond $a^\star_{\noblock-\lfloor\noblock\mgain\rfloor-\del}$ does not increase the outage diversity gain.  Equivalently, instead of having perfect CSIT, a system whose CSIT error decays with $\power$ as $\power^{- a^\star_{\noblock-\lfloor \noblock\mgain\rfloor- \del}}$ is sufficient in terms of outage diversity.   

\subsection{Predictive CSIT}
\label{sec:predictive-csit}
For systems with predictive CSIT, where mismatched CSIT of blocks up to $\block+t$ is available at the transmission of block $\block$, the achievable DMT is given as follows. 
\begin{theorem}
  \label{the:imperfect-predict} Consider transmission over the MIMO block-fading channel in \eqref{eq:channel_model} with multiplexing gain $\mgain \in [0, \mindim]$.  Assume that predictive mismatched CSIT $\incvec{\widehat\chmat}{\block+t}$ as modelled in \eqref{eq:mismatch_model} is available at the transmission of block $\block$, where $t \geq 0$ is the number of CSIT blocks predicted.  Then the optimal outage diversity is lower bounded by 
  \begin{equation}
    d(\mgain, \de) \geq \min_{\mat k \in \left\{0, \ldots, \mindim\right\}^\noblock} d_{\mat k}, 
 \end{equation}
where $d_{\mat k}$ is obtained by
\begin{equation}
  \begin{cases}
    {\rm Infimum~}& \sum_{\block=1}^\noblock\de(\maxdim-k_\block)(\mindim-k_\block) + \sum_{\block=1}^\noblock\sum_{i=1}^{\mindim-k_\block}(2i-1+\maxdim-\mindim)\overline\npfgain_{\block, i}\\
    {\rm Subject~to~}&\sum_{\block=1}^\noblock\sum_{i=1}^{\mindim-k_\block}\left(\npower_\block\left(\incvec{\widehat{\mat\npfgain}}{\block+t} \right)- \de-\overline\npfgain_{\block, i}\right)_+ + \sum_{\block=1}^\noblock k_\block \npower_{\block}(\mat k) < \noblock\mgain\\
    &\overline\npfgain_{\block, i} \geq 0, i=1, \ldots, \mindim-k_\block, 
  \end{cases}
\end{equation}
where $\npower_{\block}(\mat k) = 1+ \de\sum_{\block'=1}^{\min\{\noblock, \block+t\}} (\maxdim-k_{\block'})(\mindim-k_{\block'})$. 
\end{theorem}
\begin{proof}
  See Appendix \ref{sec:pred-csit-proof}.  
\end{proof}

The achievable DMT of  a 2-by-2 MIMO block-fading channel with $\noblock = 4$ is illustrated in Figures \ref{fig:DMT_MIMO_pred_del} and \ref{fig:DMT_MIMO_pred_va_de}.  Figure \ref{fig:DMT_MIMO_pred_del} clearly demonstrates the benefits of predicting CSIT.  With $t=0$, i.e., when only CSIT of the current transmission block is known, significant gains in outage diversity is observed, as already pointed out in \cite{KimCaire2009}.  Further gains in outage diversity are possible by increasing $t$ up to $\noblock-1$.   Note that comparing to a system with $t$ predictive blocks,  predicting $t+1$ fading blocks provides additional CSIT information only for the first $\noblock-t-1$ transmission blocks.  Therefore, the additional outage diversity gains offered by predicting $t+1$ blocks compared to that of predicting $t$ blocks decreases with $t$, as observed in Figure \ref{fig:DMT_MIMO_pred_del}.  An alternative way to improve the outage diversity is to provide better CSIT, as illustrated in Figure \ref{fig:DMT_MIMO_pred_va_de}.  In agreement with the results in \cite{KimCaire2009}, increasing $\de$ significantly improves the outage diversity.  In fact, even with $\de=1$, the outage diversity is so large that it can be considered infinite for all practical purposes.  In contrast to systems with causal CSIT, $d(\mgain, \de)$ for systems with predictive CSIT is strictly increasing with $\de$.  With $\de \to \infty$, as in systems where the CSIT quality grows exponentially with SNR, $d(\mgain, \de)$ reaches infinity, leading to a positive delay-limited capacity in many scenarios \cite{HanlyTse1998}.  

The theorem shows the impact predictive CSIT and its quality on the asymptotic outage performance, leading to essential system design guidelines.  Particularly, for systems with limited resources for channel estimation, it may be better to have high quality CSIT for a few future blocks, rather than predicting far into the future with low quality estimations.

\section{Rate-Diversity Tradeoff of Discrete Input Channels}
\label{sec:rate-divers-trad}
 In this section, we concentrate on the more practical case where the input constellation is discrete with a fixed and finite constellation $\conste$ of size $2^\consize$.  Systems with causal CSIT are studied in Section \ref{sec:causal-csit}, and then systems with predictive CSIT is studied in Section \ref{sec:acausal-csit}.  
\subsection{Causal CSIT}
\label{sec:causal-csit}
\begin{theorem}[Causal mismatched CSIT]
  \label{the:RDT_causal} Consider transmission at rate $\rate \in [0, \notx\consize]$ over the MIMO block-fading channel given in \eqref{eq:channel_model} with input constellation $\conste$ of size $2^\consize$.  Assume that mismatched CSIT $\incvec{\widehat{\chmat}}{\block-\del}$ as modelled in \eqref{eq:mismatch_model} is available at transmission block $\block$, where $\del >0$ is the delay in obtaining the CSIT.  With the long-term power constraint in \eqref{eq:pow_constraint}, the optimal RDT is given by 
  \begin{equation}
\label{eq:RDT_MIMO_causal_theorem}
    d(\rate, \de)= \norx\notx\sum_{\block=1}^{\hat \block} a_\block + \norx(d_S(\rate)- \hat\block\notx) a_{\hat\block+ 1},    
  \end{equation}
where
\begin{align}
  d_S(\rate)&= 1+ \left\lfloor \noblock\left(\notx-\frac{\rate}{\consize}\right)\right\rfloor\\
  \hat\block &= \left\lfloor \frac{d_S(\rate)}{\notx}\right\rfloor\\
  a_\block &=
  \begin{cases}
    1, &\block =1, \ldots, \del\\
    a_{\block-1} + \notx\norx\min\left\{\de, a_{\block-\del}\right\}, &\block = u+1, \ldots, \hat\block + 1.
  \end{cases}
\end{align}
\end{theorem}
\begin{proof}
  See Appendix \ref{sec:proof-theor-refth-RDTcausal}. 
\end{proof}

The optimal rate-diversity tradeoff for a MIMO block-fading channel with $\noblock =4, \notx=\norx=2$ is illustrated in Figures \ref{fig:RDT_causal_mismatch_vsdel} and \ref{fig:RDT_causal_mismatch_vsde}.   Figure \ref{fig:RDT_causal_mismatch_vsdel} shows that the optimal outage diversity increases significantly with decreasing delay in getting CSIT.   Note  that similarly to the Gaussian input case, even perfect causal CSIT may not provide gains in terms  of outage diversity over the non-CSIT case when the transmission rate is sufficiently large.  The observation can be explained from Theorem \ref{the:RDT_causal} as follows.  The outage diversity $d(\rate, \de)$ coincides with that of a system with no CSIT when the $a_b$'s in \eqref{eq:RDT_MIMO_causal_theorem} are equal to 1.  Therefore, causal CSIT provides gains in terms of outage diversity if and only if $d_S(\rate) > \del\notx$, or equivalently when 
\begin{equation}
  \rate \leq \frac{\noblock-\del}{\noblock}\consize\notx.
\end{equation}
Meanwhile, Figure \ref{fig:RDT_causal_mismatch_vsde} shows that for a given delay $\del$, significant gains  in outage diversity can be obtained by improving the quality of CSIT (increasing $\de$).  Note from Theorem \ref{the:RDT_causal} that the  outage diversity does not improve with increasing $\de$ when $\de \geq a_{\left\lceil \frac{d_S(\rate)}{\notx}\right\rceil- \del}$.  In other words the optimal outage diversity of systems with perfect CSIT can be achieved with finite $\de$.  This agrees, and generalises, the result in \cite{NguyenRasmussenGuillenLetzepis2009it}, which shows that an ARQ system with a finite number of feedback bits can achieve the optimal outage diversity of that with infinitely many feedback bits.  The result is numerically illustrated in Figure \ref{fig:RDT_causal_mismatch_vsde} for a 2-by-2 MIMO block-fading channel with $\noblock=4, \del=2$ and 16-QAM input constellation, where $\de \geq 1$ is sufficient to achieve the outage diversity of perfect CSIT systems for all transmission rates.  
\subsection{Predictive CSIT}
\label{sec:acausal-csit}
When predictive CSIT $\incvec{\widehat\chmat}{ \block+t}$, for some $t \geq 0$, is available at transmission of block $\block$, the optimal RDT is determined as follows. 
\begin{theorem}[Predictive mismatched CSIT]
\label{the:predictive-csit-rdt}
  Consider transmission at rate $\rate$ over the MIMO block-fading channel in \eqref{eq:channel_model} using input constellation $\conste$ of size $2^\consize$.   Assume that mismatched CSIT $\incvec{\widehat\chmat}{\block+t}$ as modelled in \eqref{eq:mismatch_model} is available at the transmission of block $\block$, where $t \geq 0$ is the number of future blocks with CSIT.  With the long-term power constraint in \eqref{eq:pow_constraint}, the optimal RDT is 
  \begin{equation}
    \label{eq:RDT_predictive_theorem}
    d(\rate, \de) =
    \begin{cases}
      \norx d_S(\rate)(1+ \norx d_S(\rate)\de), &t \geq \hat\block\\
      \norx\left(d_S(\rate)+ \norx\de\left(\frac{(\hat\block-t)(\hat\block+t+1)}{2}\notx^2 +d_S(\rate)\left(d_S(\rate)- \notx(\hat\block-t)\right)\right)\right), &{\rm otherwise}, 
    \end{cases}
 \end{equation}
where $d_S(\rate)= 1+ \left\lfloor \noblock\left(\notx-\frac{\rate}{\consize}\right)\right\rfloor$ and $\hat\block = \left\lfloor \frac{d_S(\rate)}{\notx}\right\rfloor$. 
\end{theorem}
\begin{proof}
  See Appendix \ref{sec:proof-theor-refth-pred-csit-rdt}. 
\end{proof}
Theorem \ref{the:predictive-csit-rdt} illustrates the impact of mismatched predictive CSIT on the outage diversity of the MIMO block-fading channel.  In contrast to the causal CSIT case, we observe from  \eqref{eq:RDT_predictive_theorem}  that the optimal outage diversity is strictly increasing with the quality of CSIT $\de$.  In effect, $d(\rate, \de) = \infty$ with perfect CSIT ($\de \to \infty$).  Moreover, for SISO systems with $t \geq 1$, or MIMO systems with $t \geq 0$, the outage curve is vertical \cite{BiglieriCaireTaricco2001,NguyenGuillenRasmussen2010it} when $\de = \infty$, leading to positive delay-limited capacity \cite{HanlyTse1998}.  The effect of $\de$ on the outage diversity of a 2-by-2 MIMO block-fading channel with $\noblock =4$ is illustrated in Figure \ref{fig:RDT_MIMO_predict_vsde}.   The figure shows that significant gains in outage diversity is obtained even for $t=0$ and relatively small $\de$, making the outage diversity effectively infinite for practical purposes, especially for small transmission rates.  

Similarly, the outage diversity is improved by increasing $t$, the number of blocks whose CSIT is available prior to transmission.  The rate-diversity tradeoff for a 2-by-2 MIMO block-fading channel with $\noblock=4$ using 16-QAM input constellation and $\de=0.5$ is illustrated in Figure \ref{fig:RDT_MIMO_predict_vsdel}.   Note from \eqref{eq:RDT_predictive_theorem} that the outage diversity cannot be further improved by increasing $t$ beyond $\hat\block= \left\lfloor \frac{d_S(\rate)}{\notx}\right\rfloor$.  This effect is illustrated in Figure \ref{fig:RDT_MIMO_predict_vsdel}, where increasing $t$ only improves the outage diversity at lower transmission rates.  At $t=\noblock-1$, CSIT of all blocks is available prior to transmission, and thus the rate-diversity tradeoff curve coincides with that of systems with full mismatched CSIT.  

Similarly to the Gaussian case in Section \ref{sec:predictive-csit}, with large $t$, increasing the number of predictive blocks leads to marginal improvement in outage diversity.  Therefore, for systems with limited resources for channel estimation, it may be more beneficial to have high quality predictive CSIT for a few future blocks, rather than having poor predictions for many future blocks.   Therefore, trading off between the number predictive blocks and the CSIT quality is required to effectively exploit the available channel-estimation resources.  
\section{Conclusions}
\label{sec:conclusion}
We have studied the asymptotic outage performance of the MIMO block-fading channel with a general model for incomplete CSIT.  The model covers a wide range of scenarios, including systems where CSIT of all fading blocks is known prior to transmission, systems with causal CSIT where a delay in obtaining CSIT is incurred, and systems with predictive CSIT where the fading gains of future fading blocks is made available at the transmitter.  The results illustrate the effects of the limited as well as imperfect  CSIT on the optimal DMT and RDT under long-term average power constraints.   The analysis reveals that the DMT (RDT) of systems with causal CSIT is limited by the delay in obtaining CSIT.  Meanwhile, the DMT(RDT) of systems with predictive CSIT is limited by the quality of CSIT.  Therefore, the quality and quantity (CSIT delay or number of predictive blocks) tradeoff is dependent on the type of CSIT available, leading to different design criteria in acquiring CSIT. 
\appendices
\section{DMT of MIMO Block-Fading Channels with Mismatched CSIT}

\label{sec:dmt-mimo-block}

\subsection{Causal CSIT-- Proof of Theorem \ref{the:imper-CSIT-MIMO-DMT}}
\label{sec:causal-csit-proof}
For a MIMO block-fading channel in \eqref{eq:channel_model}, let $\mat \eig_\block =[\eig_{\block, 1}, \ldots, \eig_{\block, \mindim}]^T$, where $0 \leq \eig_{\block, 1} \leq \ldots \leq \eig_{\block, \mindim}$ are the ordered eigenvalues  of $\chmat_\block\chmat_\block^\dag$; and let $\widehat{\mat\eig}_\block = [\widehat\eig_{\block, 1}, \ldots, \widehat\eig_{\block, \mindim}]^T$, where $0 \leq \widehat\eig_{\block, 1} \leq \ldots \leq \widehat\eig_{\block, \mindim}$ are the ordered eigenvalues of $\widehat\chmat_\block\widehat\chmat_\block^\dag$.  For an achievability result, we consider the power allocation rule $\mat\power_\block\left(\incvec{\widehat{\chmat}}{\block-\del}\right) = \power_\block\left(\incvec{\widehat{\mat\eig}}{\block-\del}\right) \ident_{\notx}$, where $\incvec{\widehat{\mat\eig}}{\block-\del} \triangleq \left(\widehat{\mat\eig}_1, \ldots, \widehat{\mat\eig}_{\block-\del}\right)$. 
 Then, the  outage probability asymptotically achieves
\begin{equation}
\label{eq:MIMO_mismatch_outage}
  \Pout(\power, \mgain\log_2\power) \doteq \Pr\left\{\sum_{\block=1}^\noblock\sum_{i=1}^\mindim \log_2\left(1+ \power_\block\left(\incvec{\widehat{\mat\eig}}{\block-\del}\right) \eig_{\block, i}\right) < \noblock\mgain\log_2 \power\right\}.
\end{equation}

Let $\widehat{\npfgain}_{\block, i} = \frac{-\log \widehat\eig_{\block, i}}{\log \power}$, following \cite{ZhengTse2003}, the distribution of $\widehat{\mat\npfgain}_\block= \left[\hat\npfgain_{\block, 1}, \ldots, \hat\npfgain_{\block, \mindim}\right]^T$ in the limit of large $\power$ is 
\begin{equation}
  f_{\widehat{\mat\npfgain}_\block} \left(\widehat{\mat\npfgain}_\block\right) =
  \begin{cases}
    \prod_{i=1}^\mindim \power^{-(2i-1+\maxdim-\mindim)\widehat\npfgain_{\block, i}},&\widehat\npfgain_{\block, 1} \geq \ldots \geq \widehat\npfgain_{\block,\mindim} \geq 0, \\
    0, &{\rm otherwise.} 
  \end{cases}
\end{equation}
Let $\incvec{\widehat{\mat\npfgain}}{\block-\del} = \left[\widehat{\mat\npfgain}_1, \ldots, \widehat{\mat\npfgain}_{\block-\del}\right]$ and $\npower_\block\left(\incvec{\widehat{\mat\npfgain}}{\block-\del}\right) = \frac{\log \power_\block\left(\incvec{\widehat{\mat\eig}}{\block-\del}\right)}{\log\power}$.  The power allocation rule asymptotically satisfies
\begin{equation}
  \int_{\incvec{\widehat{\mat\npfgain}}{\block-\del} \in \real_+^{(\block-\del)\mindim}: \widehat\npfgain_{\block', i} \geq \widehat\npfgain_{\block', i+1}} \power^{\npower_\block\left(\incvec{\widehat{\mat\npfgain}}{\block-\del}\right)} \prod_{\block'=1}^{\block-\del} \prod_{i=1}^\mindim \power^{-(2i-1+ \maxdim-\mindim) \widehat\npfgain_{\block', i}}d\incvec{\widehat{\mat\npfgain}}{\block-\del} \dot\leq \power. 
\end{equation}
Following the Varadhan's lemma \cite[Sec. 4.3]{DemboZeitouni2009}, the power constraint is asymptotically equivalent to 
\begin{equation}
  \npower_\block\left(\incvec{\widehat{\mat\npfgain}}{\block-\del}\right)  \leq 1+ \sum_{\block'=1}^{\block-\del}\sum_{i=1}^\mindim \left(2i-1+ \maxdim-\mindim\right)\widehat\npfgain_{\block', i}. 
\end{equation}
Since the outage probability is a decreasing function of transmit power, the power allocation rule with 
\begin{equation}
  \npower_{\block}\left(\incvec{\widehat{\mat\eig}}{\block-\del}\right) \equiv \npower_\block\left(\incvec{\widehat{\mat\npfgain}}{\block-\del}\right)= 1+ \sum_{\block'=1}^{\block-\del}\sum_{i=1}^\mindim \left(2i-1+ \maxdim-\mindim\right)\widehat\npfgain_{\block', i}
  \end{equation}
is optimal in terms of outage exponent. 

Therefore, letting $\npfgain_{\block, i}\triangleq \frac{-\log \eig_{\block, i}}{\log \power}$ $(\block=1, \ldots, \noblock, i=1, \ldots, \mindim)$, it follows from \eqref{eq:MIMO_mismatch_outage} that the outage probability at large SNR behaves like 
\begin{align}
  \Pout(\power, \mgain\log_2 \power) &\doteq \Pr\left\{\sum_{\block=1}^\noblock\sum_{i=1}^\mindim \log_2\left(1+ \power^{\npower_\block\left(\incvec{\widehat{\mat\npfgain}}{\block-\del}\right) - \npfgain_{\block, i}}\right) < \noblock \mgain\log_2 \power\right\} \\
&\doteq  \Pr\left\{\sum_{\block=1}^\noblock \sum_{i=1}^\mindim \left(\npower_\block\left(\incvec{\widehat{\mat\npfgain}}{\block-\del}\right)- \npfgain_{\block, i}\right)_+ < \noblock\mgain\right\}. 
\end{align}

Now for $\block=1, \ldots, \noblock$, let $\overline{\mat\eig}_\block=\left(\overline\eig_{\block, 1},  \ldots, \overline\eig_{\block, \mindim}\right)^T$, where $\overline\eig_{\block, 1} \leq \ldots \leq \overline\eig_{\block, \mindim}$ are the ordered eigenvalues of $\overline\chmat_\block\overline\chmat_\block^\dag$.  Furthermore, letting $\overline\npfgain_{\block, i} \triangleq \frac{-\log \overline\eig_{\block, i}}{\log \power}$ (for $\block=1, \ldots, \noblock,  i = 1, \ldots, \mindim$), it follows from \eqref{eq:mismatch_model} and \eqref{eq:norm_CSIT_noise} that $\overline\npfgain_{\block, i} = -\npfgain_{\block, i} + \de$.  Then the outage probability can be written as 
\begin{equation}
  \Pout(\power, \mgain\log_2 \power) \doteq \int_{\left(\incvec{\widehat{\mat\npfgain}}{\noblock}, \incvec{\overline{\mat\npfgain}}{\noblock}\right) \in \outset} \prod_{\block=1}^\noblock  f_{\overline{\mat\eig}_\block | \widehat{\mat\eig}_\block}\left(\overline{\mat\eig}_\block| \widehat{\mat\eig}_\block\right) f_{\widehat{\mat\eig}_\block}\left(\widehat{\mat\eig}_\block\right) d\incvec{\widehat{\mat\eig}}{\noblock} d\incvec{\overline{\mat\eig}}{\noblock}, 
\end{equation}
where 
\begin{equation}
\label{eq:outset_def}
  \outset \triangleq \left\{\left(\incvec{\widehat{\mat\npfgain}}{\noblock}, \incvec{\overline{\mat\npfgain}}{\noblock}\right) \in \left(\real^{\noblock\mindim}, \real^{\noblock\mindim}\right): \sum_{\block=1}^\noblock \sum_{i=1}^\mindim\left(\npower_\block\left(\incvec{\widehat{\mat\npfgain}}{\block-\del} \right)- \overline\npfgain_{\block, i} - \de \right)_+ < \noblock\mgain\right\} 
\end{equation}
is the outage set. 

Following the analysis in \cite{KimCaire2009}, the outage probability is bounded by 
\begin{align}
\label{eq:asym_bound_outage}
\Pout(\power, \mgain\log_2 \power) \dot \leq 
\int_{\outset} \prod_{\block=1}^\noblock\prod_{i=1}^\mindim &\exp\left(-\power^{-\overline\npfgain_{\block, i}}- \power^{-\left(\widehat\npfgain_{\block, i} - \de\right)} +\power^{\frac{-\overline\npfgain_{\block, i} - \left(\widehat\npfgain_{\block, i} - \de\right)}{2}}\right) \power^{-(\maxdim-\mindim)\overline\npfgain_{\block, i}} \notag\\
&\prod_{j> i} \left[\left(\power^{-\overline\npfgain_{\block, i}} - \power^{-\overline\npfgain_{\block, j}}\right)^2 \right] \power^{-(\maxdim-\mindim)\widehat\npfgain_{\block, i}}\prod_{j> i} \left[\left(\power^{-\widehat\npfgain_{\block, i}} - \power^{-\widehat\npfgain_{\block, j}}\right)^2\right]\notag\\
&\exp\left(-\power^{-\widehat\npfgain_{\block, i}}\right) \power^{-(\overline\npfgain_{\block,i}  + \widehat\npfgain_{\block, i} )}  d\incvec{\widehat{\mat\npfgain}}{\noblock} d\incvec{\overline{\mat\npfgain}}{\noblock}
\end{align}
As in \cite{KimCaire2009}, defining the $(\mindim+1)^\noblock$ disjoint integral regions
\begin{align}
 \mathcal{A}_{\mat k} \triangleq \big\{\incvec{\overline{\mat  \npfgain}}{\noblock}, \incvec{\widehat{\mat \npfgain}}{\noblock}: &\widehat\npfgain_{\block, 1} \geq \ldots\geq \widehat\npfgain_{\block, \mindim-k_\block} \geq \de > \widehat\npfgain_{\block, \mindim-k_\block+1} \geq \ldots \geq \widehat\npfgain_{\block, \mindim} \geq 0, \notag\\
&\overline\npfgain_{\block, 1} \geq \ldots \geq \overline\npfgain_{\block, \mindim-k_\block} \geq 0, \notag\\
&\overline\npfgain_{\block, \mindim-k_\block+1} = \widehat\npfgain_{\block, \mindim-k_\block+1} -\de, \ldots, \overline\npfgain_{\block, \mindim}= \widehat\npfgain_{\block, \mindim}- \de, \block =1, \ldots, \noblock
\big\}
\end{align}
where $\mat k= [k_1, \ldots, k_\noblock] \in \left\{0, \ldots, \mindim\right\}^\noblock$.  Further define the corresponding exponent $d_{\mat k}$  such that 
\begin{equation}
\int_{\outset \cap \mathcal{A}_{\mat k}} \prod_{\block=1}^\noblock f_{\overline{\mat \npfgain}_\block|\widehat{\mat\npfgain}_\block}\left(\overline{\mat\npfgain}_\block|\widehat{\mat\npfgain}_\block\right)f_{\widehat{\mat\npfgain}_\block}\left(\widehat{\mat\npfgain}_\block\right) d\overline{\mat\npfgain}_\block d\widehat{\mat \npfgain}_\block \doteq \power^{-d_{\mat k}}. 
\end{equation}
Then, the outage diversity at multiplexing gain $\mgain$ satisfies
\begin{equation}
  d(\mgain, \de) \geq \min_{\mat k} \left\{d_{\mat k}\right\}. 
\end{equation}
We now have that \cite{KimCaire2009}
\begin{align}
\int_{\outset \cap \mathcal{A}_{\mat k}} \prod_{\block=1}^\noblock f_{\overline{\mat \npfgain}_\block|\widehat{\mat\npfgain}_\block}\left(\overline{\mat\npfgain}_\block|\widehat{\mat\npfgain}_\block\right)f_{\widehat{\mat\npfgain}_\block}\left(\widehat{\mat\npfgain}_\block\right) d\overline{\mat\npfgain}_\block d\widehat{\mat \npfgain}_\block &\doteq\notag\\
&\hspace{-2 in} \int_{\outset \cap \mathcal{A}_{{\mat k}}} \prod_{\block=1}^\noblock \power^{-\sum_{i=1}^{\mindim-k_\block} (2i-1+\maxdim-\mindim)\overline\npfgain_{\block, i}} \power^{-\sum_{i=1}^\mindim (2i-1 + \maxdim-\mindim)\widehat\npfgain_{\block, i}} d\incvec{\overline{\mat\npfgain}}{\noblock} d\incvec{\widehat{\mat\npfgain}}{\noblock}. 
\end{align}
Together with \eqref{eq:outset_def}, we have that
\begin{align}
  d_{\mat k} &= \inf_{\mathcal{A}_{\mat k}}  \left\{\sum_{\block=1}^\noblock \sum_{i=1}^{\mindim-k_\block} \left(2i-1+\maxdim-\mindim\right)\overline\npfgain_{\block, i} + \sum_{i=1}^\mindim (2i-1+\maxdim-\mindim)\widehat\npfgain_{\block, i}\right\}\\
\label{eq:part-diver}&{\rm s.t.} \sum_{\block=1}^\noblock \sum_{i=1}^{\mindim-k_\block} \left(\npower_\block\left(\incvec{\widehat{\mat\npfgain}}{\block-\del}\right) - \de - \overline\npfgain_{\block, i}\right)_+ + \sum_{i=\mindim-k_\block+1}^\mindim \left(\npower_\block\left(\incvec{\widehat{\mat \npfgain}}{\block-\del}\right) - \widehat\npfgain_{\block, i}\right)_+ < \noblock\mgain,
\end{align}
where we recall that $\npower_\block\left(\incvec{\widehat{\mat \npfgain}}{\block- \del}\right) = 1+ \sum_{\block'=1}^{\block-\del} \sum_{j=1}^\mindim (2j-1+\maxdim-\mindim)\widehat\npfgain_{\block', i}$.  For $\widehat\npfgain_{\block, i}$ with $i \leq \mindim-k_\block$, decreasing $\widehat\npfgain_{\block,i}$ decreases the objective function, while the constraint is unchanged.  Therefore, the optimiser satisfies $\widehat\npfgain_{\block, i} = \de, i=1, \ldots, \mindim-k_\block$.  Thus, $d_{\mat k}$ is obtained from
\begin{equation}
\begin{cases}
{\rm Infimum~} &\noblock\maxdim\mindim\de+ \sum_{\block=1}^\noblock\sum_{i=1}^\mindim(2i-1+\maxdim-\mindim) \overline\npfgain_{\block, i} \\
{\rm Subject ~ to~}& \sum_{\block=1}^\noblock\sum_{i=1}^\mindim \left(\npower_\block\left(\incvec{\widehat{\mat\npfgain}}{\block-\del}\right)- \de-\overline\npfgain_{\block, i}\right)_+ < \noblock\mgain\\
&\overline\npfgain_{\block, 1} \geq \ldots\geq \overline\npfgain_{\block,\mindim-k_\block} \geq 0 \geq \overline\npfgain_{\block, \mindim-k_\block+1} \geq \ldots\geq \overline\npfgain_{\block, \mindim} \geq -\de
\end{cases}
\end{equation}
where $\npower_\block\left(\incvec{\widehat{\mat\npfgain}}{\block-\del}\right)= 1+\maxdim\mindim\de(\block-\del)_+ + \sum_{\block'=1}^{\block-\del}\sum_{i=\mindim-k_\block'+1}^\mindim(2i-1+\maxdim-\mindim)\overline\npfgain_{\block,i}$.  \endproof
\subsection{Causal CSIT in Vector Channels-- Proof of Corollary \ref{cor:DMT_SISO_causal_mismatch}}
\label{sec:causal-csit-vector}
  With $\mindim =1$, letting $a_\block = \overline\npfgain_{\block, 1} + \de$. Noting that the solution of \eqref{eq:opt_dmt_mimo_causal} satisfies $-\de \leq \overline\npfgain_{\block, 1} \leq \npower_\block\left(\incvec{\widehat{\mat\npfgain}}{\block-\del}\right) - \de$, a lower bound of $d(\mgain, \de)$ is obtained from Theorem \ref{the:imper-CSIT-MIMO-DMT} as
 \begin{equation}
\label{eq:SIMO_opt}
    \begin{cases}
      {\rm Infimum}& \maxdim\sum_{\block=1}^\noblock a_\block\\
      {\rm Subject~to~}&\sum_{\block=1}^\noblock\left(\npower_\block(\mat a) - a_\block\right) < \noblock\mgain\\
      & \npower_{\block}(\mat a) \geq a_\block \geq 0, \block=1, \ldots, \noblock,
    \end{cases}
  \end{equation}
where $\npower_\block(\mat a) = 1+ \maxdim\sum_{\block'=1}^{\block-\del} \min\left\{a_\block, \de\right\}$. The constraints in \eqref{eq:SIMO_opt} are equivalent to
\begin{equation}
\label{eq:simp_constraint}
\begin{cases}
 & \sum_{\block=1}^\noblock  a_\block> \sum_{\block=1}^\noblock \npower_\block(\mat a) - \noblock\mgain. \\
 &\npower_\block(\mat a) \geq a_\block \geq 0, \block=1, \ldots, \noblock.
\end{cases}
\end{equation}
 We consider the following cases. 
\begin{itemize}
\item When $\noblock-\del-\noblock\mgain \leq 0$, since $\npower_{\block}(\mat a) \geq 1$, it follows from \eqref{eq:simp_constraint} that 
  \begin{equation}
  \sum_{\block=1}^\noblock \npower_\block(\mat a) - \noblock\mgain \geq  \noblock(1-\mgain), 
  \end{equation}
with equality attained when, for e.g., 
\begin{equation}
  a_\block=
  \begin{cases}
    0, &\block=1, \ldots, \left\lfloor \noblock\mgain\right\rfloor\\
    \frac{\noblock(1-\mgain)}{\noblock-\lfloor\noblock\mgain\rfloor}, &\block=\lfloor \noblock\mgain\rfloor + 1, \ldots, \noblock
  \end{cases}
\end{equation}
satisfying the constraints in \eqref{eq:simp_constraint}.  Therefore, $d(\mgain, \de) \geq \maxdim\noblock(1-\mgain)$ in this case. 
\item When $\noblock-\del-\noblock\mgain > 0$, the right-hand-side of the first constraint in \eqref{eq:simp_constraint} is minimised when, for e.g., 
  \begin{equation}
    a_\block=
    \begin{cases}
      0, &\block=1, \ldots, \lfloor \noblock\mgain\rfloor\\
      1-\noblock\mgain+\lfloor\noblock\mgain\rfloor, &\block = \lfloor \noblock\mgain\rfloor+1\\
      \npower_{\block}(\mat a), &\block=\lfloor \noblock\mgain\rfloor+2, \ldots, \noblock. 
    \end{cases}
  \end{equation}
  satisfying the constraints in~\eqref{eq:SIMO_opt}.  By letting $a^\star_i = a_{i+\lfloor\noblock\mgain\rfloor}, i=1, \ldots, \noblock-\lfloor\noblock\mgain\rfloor$, noting that $a_\block=a_{\block-1} + \maxdim\min\{a_{\block-\del}, \de\}, \block \geq \del$, we arrive at \eqref{eq:recursive_a}.
\end{itemize}
Combining the two cases, we arrive at \eqref{eq:mismatch_DMT_causal_SISO}. \endproof
\subsection{Predictive CSIT-- Proof of Theorem \ref{the:imperfect-predict}}
\label{sec:pred-csit-proof}
As in Appendix \ref{sec:causal-csit-proof},  for channel with imperfect CSIT $\incvec{\widehat\chmat}{\min\{\noblock,\block+t\}}$, for large $\power$, the power allocation  rule asymptotically satisfies
\begin{equation}
\label{eq:pred_mimo_dmt_npower}
  \npower_\block\left(\incvec{\widehat{\mat\npfgain}}{\block+t}\right) = 1+ \sum_{\block'=1}^{\min\{\noblock, \block+t\}}\sum_{i=1}^\mindim (2i-1+\maxdim-\mindim)\widehat\npfgain_{\block, i}.  
\end{equation}
Therefore, following similar arguments to Appendix \ref{sec:causal-csit-proof},  the similar to the arguments in Section \ref{sec:causal-csit-proof}, the outage diversity  is lower bounded by 
\begin{equation}
  d(\mgain, \de) \geq \min_{\mat k \in  \left\{0, \ldots, \mindim\right\}^\noblock} d_{\mat k}.
\end{equation}
In which, $d_{\mat k}$ is defined as
\begin{equation}
d_{\mat k} = \inf_{\left(\incvec{\overline{\mat\npfgain}}{\noblock}, \incvec{\widehat{\mat\npfgain}}{\noblock}\right) \in \mathcal{A}_{\mat k} \cap \outset} \left\{\sum_{\block=1}^\noblock\sum_{i=1}^{\mindim-k_\block} (2i-1+\maxdim-\mindim) \overline\npfgain_{\block, i} + \sum_{i=1}^\mindim(2i-1+\maxdim-\mindim)\widehat\npfgain_{\block, i}\right\}, 
\end{equation}
where 
\begin{align}
  \outset &= \left\{\left(\incvec{\overline{\mat\npfgain}}{\noblock}, \incvec{\widehat{\mat\npfgain}}{\noblock}\right): \sum_{\block=1}^\noblock \sum_{i=1}^{\mindim-k_\block}\left(\npower_{\block}\left(\incvec{\widehat{\mat \npfgain}}{\block+t}\right) - \de-\overline\npfgain_{\block, i}\right)_+ + \sum_{i=\mindim-k_\block+1}^\mindim  \left(\npower_\block\left(\incvec{\widehat{\mat\npfgain}}{\block+t}\right)- \widehat\npfgain_{\block, i}\right)_+ < \noblock\mgain\right\} 
\end{align}
and $\npower_\block\left(\incvec{\widehat{\mat\npfgain}}{\block+t}\right)$ is defined in \eqref{eq:pred_mimo_dmt_npower}.  In this case, noting that decreasing $\widehat\npfgain_{\block, i}$ decreases both the objective function and the constraint in $\outset$.  Therefore, the optimiser satisfies 
\begin{equation}
  \widehat\npfgain_{\block, i} =
  \begin{cases}
    \de, &i \leq \mindim-k_\block\\
    0, &{\rm otherwise.}
  \end{cases}
\end{equation}
and thus $\npower_{\block}\left(\incvec{\widehat{\mat\npfgain}}{\block+t}\right) = 1+ \de\sum_{\block'=1}^{\min\{\noblock,\block+t\}} (\maxdim-k_{\block'})(\mindim-k_{\block'})$.   It follows that $d_{\mat k}$ is obtained from
\begin{equation}
  \begin{cases}
    {\rm Infimum~}& \sum_{\block=1}^\noblock\de(\maxdim-k_\block)(\mindim-k_\block) + \sum_{\block=1}^\noblock\sum_{i=1}^{\mindim-k_\block}(2i-1+\maxdim-\mindim)\overline\npfgain_{\block, i}\\
    {\rm Subject~to~}&\sum_{\block=1}^\noblock\sum_{i=1}^{\mindim-k_\block}\left(\npower_\block\left(\incvec{\widehat{\mat\npfgain}}{\block+t} \right)- \de-\overline\npfgain_{\block, i}\right)_+ + \sum_{\block=1}^\noblock k_\block \npower_{\block}\left(\incvec{\widehat{\mat \npfgain}}{\block+t} \right)< \noblock\mgain\\
    &\npfgain_{\block, i} \geq 0, i=1, \ldots, \mindim-k_\block
  \end{cases}
\end{equation}
as required. \endproof

\section{RDT of MIMO Block-Fading Channels with Mismatched CSIT}
\subsection{Proof of Theorem \ref{the:RDT_causal}}
\label{sec:proof-theor-refth-RDTcausal}
Let $\fgain_{\block,\tx,\rx}$ be the fading gain corresponding to the link between transmit antenna $\tx$ and receive antenna $\rx$ in block $\block$.  Similarly, let $\widehat\fgain_{\block,\tx,\rx}$ be the estimate of $\fgain_{\block,\tx,\rx}$ and $\overline\fgain_{\block,\tx,\rx}$ be the elements of $\overline\chmat_\block$ defined in \eqref{eq:norm_CSIT_noise}.  Further define $\pfgain_{\block,\tx,\rx} \triangleq |\fgain_{\block,\tx,\rx}|^2, \widehat\pfgain_{\block,\tx,\rx} \triangleq |\widehat\fgain_{\block,\tx,\rx}|$ and $\overline\pfgain_{\block,\tx,\rx} \triangleq |\overline\fgain_{\block,\tx,\rx}|^2$ be the corresponding power fading gains.  For an achievability proof, let us consider the power allocation rule 
\begin{equation}
  \mat\power_\block\left(\incvec{\chmat}{\block-\del}\right) = \power\left(\incvec{\widehat\mpfgain}{\block-\del}\right) \ident_\notx, \; \block = 1, \ldots, \noblock, 
\end{equation}
where $\widehat\mpfgain \in \real_+^{\noblock\times\notx\times\norx}$ is the power fading gain matrix in block $\block$ with entries $\widehat\pfgain_{\block, \tx,\rx}$.  Define $\npfgain_{\block,\tx,\rx}\triangleq \frac{-\log \pfgain_{\block,\tx,\rx}}{\log \power}, \widehat\npfgain_{\block,\tx,\rx} = \frac{-\log \widehat\pfgain_{\block,\tx,\rx}}{\log \power}$ and $\overline\npfgain_{\block,\tx,\rx}= \frac{-\log\overline\pfgain_{\block,\tx,\rx}}{\log \power}$.  It then follows from \eqref{eq:mismatch_model} and \eqref{eq:norm_CSIT_noise} that $\overline\npfgain_{\block,\tx,\rx}= \npfgain_{\block,\tx,\rx} - \de$.  Let $\mnpfgain_\block, \widehat\mnpfgain_\block, \overline\mnpfgain_\block \in \real^{\notx\times\norx}$ be the matrices with entries $\npfgain_{\block,\tx,\rx}, \widehat\npfgain_{\block,\tx,\rx}$ and $\overline\npfgain_{\block,\tx,\rx}$ correspondingly.  Further define $\npower_\block\left(\incvec{\widehat\mpfgain}{\block-\del}\right) = \npower_\block\left(\incvec{\widehat\mnpfgain}{\block-\del}\right) \triangleq \frac{-\log \power\left(\incvec{\widehat\mpfgain}{\block-\del}\right)}{\log \power}$.  Since that the distribution of $\widehat\npfgain_{\block, \tx, \rx}$ in the limit of large $\power$ is
\begin{equation}
  f_{\widehat\npfgain_{\block, \tx, \rx}}(\widehat\npfgain_{\block, \tx, \rx})   \doteq
  \begin{cases}
    \power^{-\widehat\npfgain_{\block, \tx,\rx} }, &\widehat\npfgain_{\block,\tx,\rx} \geq 0\\
    0, &{\rm otherwise, }
  \end{cases}
\end{equation}
it follows from the arguments in Appendix \ref{sec:causal-csit-proof} that the power allocation rule with 
\begin{equation}
   \npower_\block\left(\incvec{\widehat\mnpfgain}{\block-\del}\right) = 1+ \sum_{\block'=1}^{\block-\del}\sum_{\tx=1}^\notx\sum_{\rx=1}^\norx \widehat\npfgain_{\block',\tx,\rx}
\end{equation}
is optimal asymptotically.  

Following \eqref{eq:pout} and \eqref{eq:MIMO_info},  in the limit of large $\power$, the outage probability is given by 
\begin{align}
  \Pout(\power, \rate) &\doteq \Pr\left\{\frac{1}{\noblock}\sum_{\block=1}^\noblock\infoX\left(\chmat_\block\power^\half_\block\left(\incvec{\widehat\mnpfgain}{\block-\del}\right)\right) < \rate\right\}, \\
&\doteq \Pr\left\{\sum_{\block=1}^\noblock T\left(\incvec{\widehat\mnpfgain}{\block-\del}, \overline\mnpfgain_\block\right)> 2^{\consize\notx}\noblock(\consize\notx-\rate) \right\}  \label{eq:outprob_mimo_causal}
\end{align}
where 
\begin{align}
  T\left(\incvec{\widehat\mnpfgain}{\block-\del},\overline\mnpfgain_\block\right) &\triangleq \sum_{\txsig \in \conste^\notx} \expectation{\boldsymbol z}{\log_2\left(\sum_{\txsig' \in \conste^\notx} \exp\left(-\left\|\chmat_\block\power_\block^{\half}\left(\incvec{\widehat\mnpfgain}{\block-\del}\right)(\txsig-\txsig') + \boldsymbol z\right\|^2 + \|\boldsymbol z\|^2\right)\right)}\notag\\
&\hspace{-1 in}=\sum_{\txsig \in \conste^\notx} \expectation{\mat z}{\log_2\left(\sum_{\txsig' \in \conste^\notx} \exp\left(\sum_{\rx=1}^\norx -\left|\sum_{\tx=1}^\notx \power^{\frac{\npower_\block\left(\incvec{\widehat\mnpfgain}{\block-2}\right) - \overline\npfgain_{\block, \tx,\rx} - \de}{2}} e^{\im\ang_{\block, \tx,\rx}}(x_\tx- x_{\tx}')\right|^2 + |z_\rx|^2\right)\right)}.
\end{align}
Here $z_\rx$ is the $\rx$th entry of $\mat z$ and $\im = \sqrt{-1}$ is the imaginary number unit. 

For any $\epsilon > 0$, let $\goodset_\block  \triangleq \bigcup_{\rx=1}^\norx \goodset_{\block,\rx}$ and $\kappa_\block = \left|\goodset_\block\right|$, where 
\begin{equation}
  \goodset_{\block,\rx} = \left\{\tx \in \left\{1, \ldots, \notx\right\}: \overline\npfgain_{\block,\tx,\rx} + \de < \npower_\block\left(\incvec{\widehat\mnpfgain}{\block-\del}\right)- \epsilon\right\}
\end{equation}
For any $\rx  \in \left\{1, \ldots, \norx\right\}$, let $\overline\npfgain_{\block, \rx} = \max_{\tx \in \goodset_{\block, \rx}} \left\{\overline\npfgain_{\block, \tx, \rx}\right\}$.  If there exists $\tx \in \goodset_{\block, \rx}$ such that $x_\tx \neq x'_\tx$ then
\begin{align}
  \lim_{\power \to \infty}   \left|\sum_{\tx=1}^\notx \power^{\frac{\npower_\block\left(\incvec{\widehat\mnpfgain}{\block-2}\right) - \overline\npfgain_{\block, \tx,\rx} - \de}{2}} e^{\im\ang_{\block, \tx,\rx}}(x_\tx- x_{\tx}')\right|& \notag\\
&\hspace{-2 in}= \lim_{\power \to \infty} \left| \power^{\frac{\npower_\block\left(\incvec{\hat\mnpfgain}{\block-\del}\right) - \overline\npfgain_{\block, \rx} - \de}{2}}\sum_{\tx=1}^\notx\power^{\frac{\overline\npfgain_{\block, \rx}-\overline\npfgain_{\block, \tx, \rx} }{2}} e^{\im\ang_{\block,\tx,\rx}}(x_\tx- x'_\tx)\right|\notag\\
&\hspace{-2 in} \geq \lim_{\power\to\infty}\left| \power^{\frac{\npower_\block\left(\incvec{\widehat\mnpfgain}{\block-\del}\right)- \overline\npfgain_{\block,\rx} - \de}{2}} \sum_{\tx \in \goodset_{\block,\rx}} e^{\im\ang_{\block,\tx,\rx}}(x_\tx-x'_\tx)\right| = \infty
\end{align}
with probability 1 since $\ang_{\block,\tx,\rx}$ are uniformly distributed in $[-\pi, \pi]$.  Thus, for asymptotically large $\power$, 
\begin{equation}
  T\left(\incvec{\widehat\mnpfgain}{\block-\del}, \overline\mnpfgain_{\block}\right) \dot\leq \sum_{\txsig \in \conste^\notx} \log_2\left(\sum_{\txsig' \in\conste^\notx} \openone{x_\tx= x'_\tx, \forall \tx \in \goodset_\block}\right)= 2^{\consize\notx}\consize\left(\notx-\kappa_\block\right). 
\end{equation}
Therefore, it follows from \eqref{eq:outprob_mimo_causal} that  
\begin{align}
  \Pout(\power, \rate) &\dot\leq \Pr\left( \sum_{\block= 1}^\noblock \kappa_\block < \frac{\noblock\rate}{\consize}\right)\\
&\doteq \int_{\left(\incvec{\widehat\mnpfgain}{\noblock}, \incvec{\overline\mnpfgain}{\noblock}\right) \in \outset} \prod_{\block, \tx, \rx} f_{\overline\npfgain_{\block,\tx, \rx}| \widehat\npfgain_{\block, \tx,\rx}} \left(\overline\npfgain_{\block, \tx, \rx}| \widehat\npfgain_{\block, \tx,\rx}\right) d\overline\npfgain_{\block,\tx,\rx}  d \widehat\npfgain_{\block, \tx,\rx}
\end{align}
where 
\begin{align}
  \outset&= \left\{\left(\incvec{\overline\mnpfgain}{\noblock}, \incvec{\widehat\mnpfgain}{\noblock}\right) \in \real^{2\noblock\norx\notx}: \sum_{\block=1}^\noblock  \kappa_\block < \frac{\noblock\rate}{\consize}\right\}\\
&=\left\{\left(\incvec{\overline\mnpfgain}{\noblock}, \incvec{\widehat\mnpfgain}{\noblock}\right): \sum_{\block=1}^\noblock\sum_{\tx=1}^\notx \openone{\overline\npfgain_{\block,\tx} + \de < \npower_\block\left(\incvec{\widehat\mnpfgain}{\block-\del}\right)  - \epsilon} < \frac{\noblock\rate}{\consize}\right\}.
\end{align}
Here $\overline\npfgain_{\block,\tx}= \min\left\{\overline\npfgain_{\block,\tx,\rx}, \rx=1, \ldots, \norx\right\}$. 
Following the arguments in \cite{KimNguyenGuillen2009}, the outage diversity is bounded by 
\begin{equation}
\label{eq:1stinf_RDT_MIMO}  d(\rate, \de) \geq \inf_{\left(\incvec{\overline\mnpfgain}{\noblock}, \incvec{\widehat\mnpfgain}{\noblock}\right) \in \overline\outset}\left\{\sum_{(\block, \tx, \rx): -\de \leq \overline\npfgain_{\block, \tx, \rx}= \widehat\npfgain_{\block, \tx, \rx} - \de < 0} \widehat\npfgain_{\block, \tx, \rx} + \sum_{(\block,\tx,\rx): \overline\npfgain_{\block, \tx, \rx} \geq 0, \widehat\npfgain_{\block,\tx,\rx} \geq \de} (\overline\npfgain_{\block,\tx,\rx}  + \widehat\npfgain_{\block, \tx, \rx})\right\}, 
\end{equation}
where 
\begin{equation}
\label{eq:all_constraint_RDT}
\overline\outset \triangleq \left\{\left(\incvec{\overline\mnpfgain}{\noblock}, \incvec{\widehat\mnpfgain}{\noblock}\right) \in \outset: \left\{-\de \leq \overline\npfgain_{\block,\tx,\rx}= \widehat\npfgain_{\block,\tx,\rx} - \de < 0\right\} {\rm ~or~} \left\{\overline\npfgain_{\block,\tx,\rx} \geq 0, \widehat\npfgain_{\block,\tx,\rx} \geq \de\right\}
\right\}
\end{equation}

Noting from \eqref{eq:all_constraint_RDT} that for $(\overline\npfgain_{\block,\tx,\rx},\widehat\npfgain_{\block, \tx, \rx}) \in \overline\outset$, for $(\block, \tx,\rx)$ such that $\widehat\npfgain_{\block,\tx,\rx} \geq \de$, decreasing $\widehat\npfgain_{\block,\tx,\rx}$ decreases the objective function in \eqref{eq:1stinf_RDT_MIMO}, while the constraint in \eqref{eq:all_constraint_RDT} is relaxed.  Therefore, together with the constraint in \eqref{eq:all_constraint_RDT}, the solution of \eqref{eq:1stinf_RDT_MIMO} satisfies
\begin{equation}
  \widehat\npfgain_{\block,\tx, \rx} = \min\left\{\overline\npfgain_{\block,\tx,\rx} + \de, \de\right\}
\end{equation}
Therefore, letting $a_{\block,\tx, \rx} = \overline\npfgain_{\block,\tx,\rx} + \de$ and $a_{\block, \tx} = \min\left\{a_{\block, \tx, \rx}, \rx =1, \ldots, \norx\right\}$, it follows that
\begin{equation}
\label{eq:inf_RDT_MIMO}
  d(\rate, \de) \geq \inf_{\mat a \in \widehat \outset} \left\{\sum_{\block =1}^\noblock \sum_{\tx=1}^\notx \sum_{\rx=1}^\norx a_{\block, \tx, \rx} \right\}, 
\end{equation}
where
\begin{equation}
  \widehat{\outset} \triangleq \left\{\mat a \in \real_+^{\noblock\times\notx\times\norx}: \sum_{\block=1}^\noblock\sum_{\tx=1}^\notx \openone{a_{\block, \tx} < \overline\npower_\block(\mat a) - \epsilon}  < \frac{\noblock\rate}{\consize}\right\}
\end{equation}
and $\overline\npower_{\block}(\mat a) = 1+ \sum_{\block'=1}^{\block-\del}\sum_{\tx=1}^\notx\sum_{\rx=1}^\norx \min\{a_{\block, \tx, \rx}, \de\}$.
The infimum in \eqref{eq:inf_RDT_MIMO} is achievable when, for $\rx=1, \ldots, \norx$, 
\begin{equation}
  a_{\block, \tx, \rx} =
  \begin{cases}
    0, & \noblock(\block-1)+ \tx < \frac{\noblock\rate}{\consize}\\
    \overline\npower_\block(\mat a) - \epsilon, &{\rm otherwise.}
  \end{cases}
\end{equation}
Therefore, letting $\hat \block \triangleq \left\lfloor \frac{d_S(\rate)}{\notx}\right\rfloor \triangleq \left\lfloor\frac{1+ \left\lfloor \noblock\left(\notx-\frac{\rate}{\consize}\right)\right\rfloor}{\notx}\right\rfloor$, the outage diversity is lower bounded by 
\begin{equation}
\label{eq:bound_diver}
  d(\rate, \de) \geq \norx \notx\sum_{\block=1}^{\hat\block} a'_\block + \norx\left(d_S(\rate)- \hat\block\notx\right) a'_{1+\hat\block}, 
\end{equation}
where for $\rx=1, \ldots, \norx$
\begin{equation}
  a'_\block = 
  \begin{cases}
    1-\epsilon, &\block = 1, \ldots, \del\\
    a'_{\block-1} + \notx\norx \min\{\de, a'_{\block-u}\} -\epsilon, &\block = u+1, \ldots, \hat\block+1. 
  \end{cases}
\end{equation}
By letting $\epsilon \downarrow 0$, the outage diversity is lower bounded by  \eqref{eq:RDT_MIMO_causal_theorem}. 

  On the other hand, using the genie-aided arguments as in \cite{NguyenThesis2009,KimNguyenGuillen2009}, the outage diversity is upper bounded by that of a channel consisting of $\notx$ parallel channels, each is a block-fading channel with $\norx$ receive antenna.  Using similar approach as in the previous part of the proof,  the outage diversity of the genie-aided channel is  also given by ~\eqref{eq:RDT_MIMO_causal_theorem}.  This concludes the proof of the Theorem. \endproof
\subsection{Proof of Theorem \ref{the:predictive-csit-rdt}}
\label{sec:proof-theor-refth-pred-csit-rdt}
Following the arguments in Appendix \ref{sec:proof-theor-refth-RDTcausal}, the optimal outage diversity of a MIMO block-fading channel with predictive CSIT $\incvec{\widehat\chmat}{\block+t}$  and mismatched CSIT exponent $\de$ is given by\footnote{As in Appendix \ref{sec:proof-theor-refth-RDTcausal}, the final result is obtained by letting $\epsilon \to 0$, which does not affect the analysis, hence $\epsilon$ is removed for simplicity.}
\begin{equation}
\label{eq:inf_rdt_mimo_predict}
  d(\rate, \de) = \inf_{\mat a \in \widehat\outset}\left\{\sum_{\block=1}^\noblock \sum_{\tx=1}^\notx \sum_{\rx=1}^\norx a_{\block, \tx, \rx}\right\}, 
\end{equation}
where 
\begin{align}
  \widehat\outset &\triangleq \left\{\mat a \in \real_+^{\noblock\notx}: \sum_{\block=1}^\noblock\sum_{\tx=1}^\notx \openone{a_{\block,\tx} < \overline\npower_\block(\mat a) } < \frac{\noblock\rate}{\consize}\right\}\\
\overline\npower_\block(\mat a) &\triangleq 1+\sum_{\block'=1}^{\min\left\{\block+t, \noblock\right\}}\sum_{\tx=1}^\notx \sum_{\rx=1}^\norx\min\left\{a_{\block', \tx, \rx}, \de\right\}. 
\end{align}
For any $\mat a \in \widehat \outset$, there are $d_S(\rate)$ coefficients $a_{\block, \tx}$'s satisfying $a_{\block,\tx} \geq \overline\npower_\block(\mat a)$.  The infimum in \eqref{eq:inf_rdt_mimo_predict} is therefore attained when, for $\rx=1, \ldots, \norx$, 
\begin{equation}
  a^\star_{\block, \tx, \rx} =
  \begin{cases}
    \overline\npower_{\block}(\mat a^\star), &(\block-1)\notx+ \tx \leq d_S(\rate)\\    
    0, &{\rm otherwise}. 
  \end{cases}
\end{equation}
Then, for $\block, \tx$ such that $(\block-1) \notx + \tx \leq d_S(\rate)$, 
\begin{equation}
  a^\star_{\block,\tx} \geq  1+ \notx\norx\min\{a^\star_{\block, \tx}, \de\} \geq \de.
\end{equation}
Thus,  when $(\block-1) \notx+\tx  \leq d_S(\rate)$,
\begin{align}
  a^\star_{\block, \tx, \rx}  &= 1+  \sum_{\block'=1}^{\block+t}\sum_{\tx=1}^\notx \sum_{\rx=1}^\norx\min\left\{a^\star_{\block, \tx, \rx}, \de\right\}\\
  &=
  \begin{cases}
    1+ \notx\norx(\block+t)\de, &\block+t \leq \left\lfloor \frac{d_S(\rate)}{\notx}\right\rfloor\\
    1+ \norx d_S(\rate)\de, &{\rm otherwise}. 
  \end{cases}
\end{align}
Letting $\hat\block = \left\lfloor \frac{d_S(\rate)}{\notx}\right\rfloor$.  If $\hat\block \leq t$,  the optimal outage diversity is 
\begin{equation}
  d(\rate, \de)  = \sum_{\block=1}^\noblock\sum_{\tx=1}^\notx\sum_{\rx=1}^\norx a^\star_{\block, \tx, \rx} = \norx d_S(\rate)\left(1+ \norx d_S(\rate)\de\right). 
\end{equation}
Meanwhile, if $\hat\block > t$, the optimal outage diversity is 
\begin{align}
  d(\rate, \de) &=  \sum_{\block=1}^\noblock\sum_{\tx=1}^\notx\sum_{\rx=1}^\norx a^\star_{\block,\tx} \notag\\
&= \norx\left(\sum_{\block=1}^{\hat\block-t} \sum_{\tx=1}^\notx \left(1+ \notx\norx(\block+t) \de\right)+ \left(d_S(\rate) - \notx(\hat\block-t)\right)(1+ \norx d_S(\rate) \de)\right) \notag\\
&=\norx \left(d_S(\rate) + \norx\de \left(\frac{(\hat\block-t)(\hat\block+t+1)}{2}\notx^2 + d_S(\rate) (d_S(\rate)- \notx(\hat\block-t))\right)\right)\notag 
\end{align}
as required.  \endproof

\bibliographystyle{IEEEtran}
\bibliography{bib_database}
\newpage
%
\newpage
\begin{figure}[htbp]
  \centering
  \includegraphics[width=1\columnwidth]{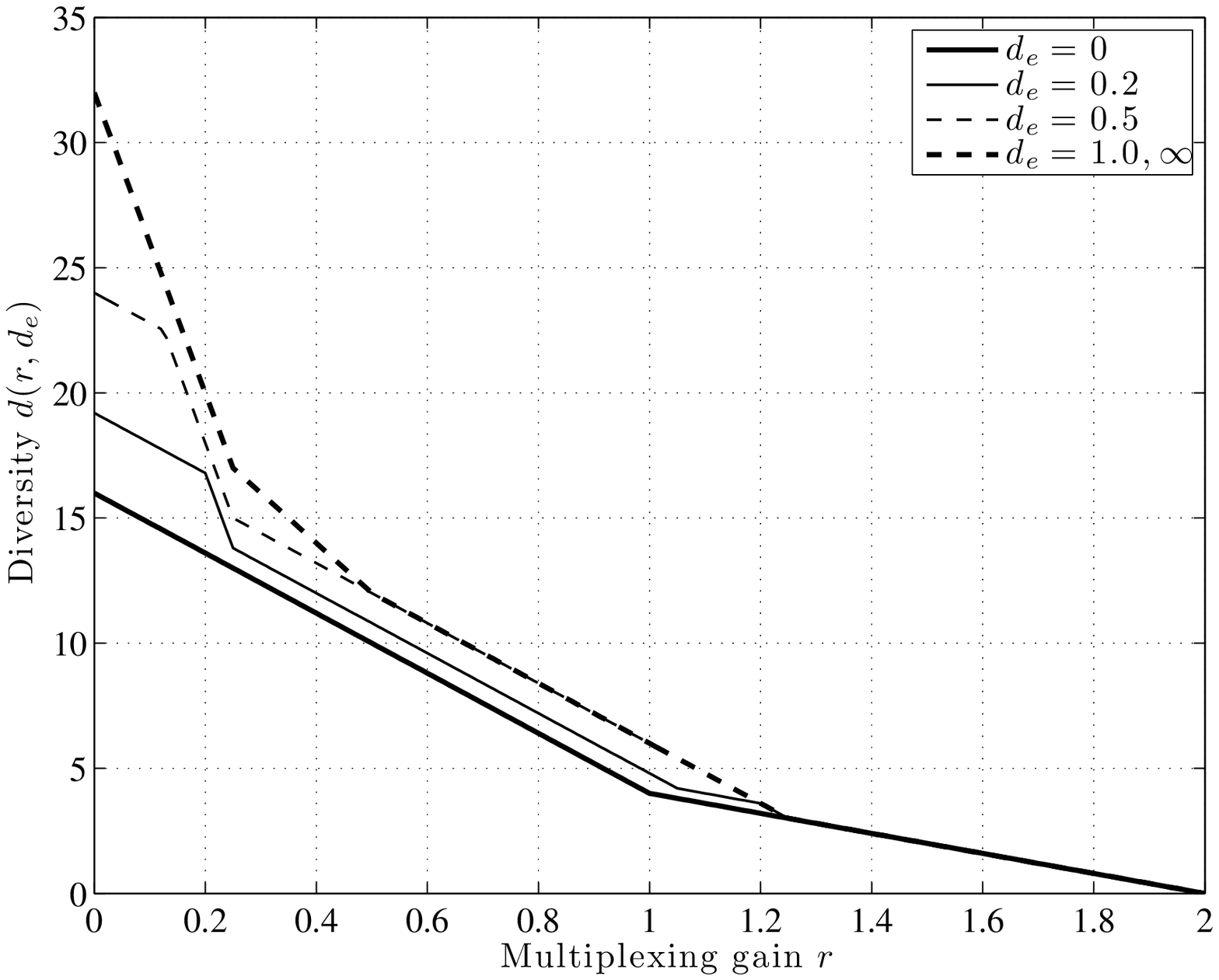}
  \caption{The achievable DMT of a 2-by-2 MIMO block-fading channel with $\noblock=4$ causal CSIT $\del=3$.}
  \label{fig:DMT_MIMO_u3}
\end{figure}
\newpage
\begin{figure}[htbp]
  \centering
  \includegraphics[width=1\columnwidth]{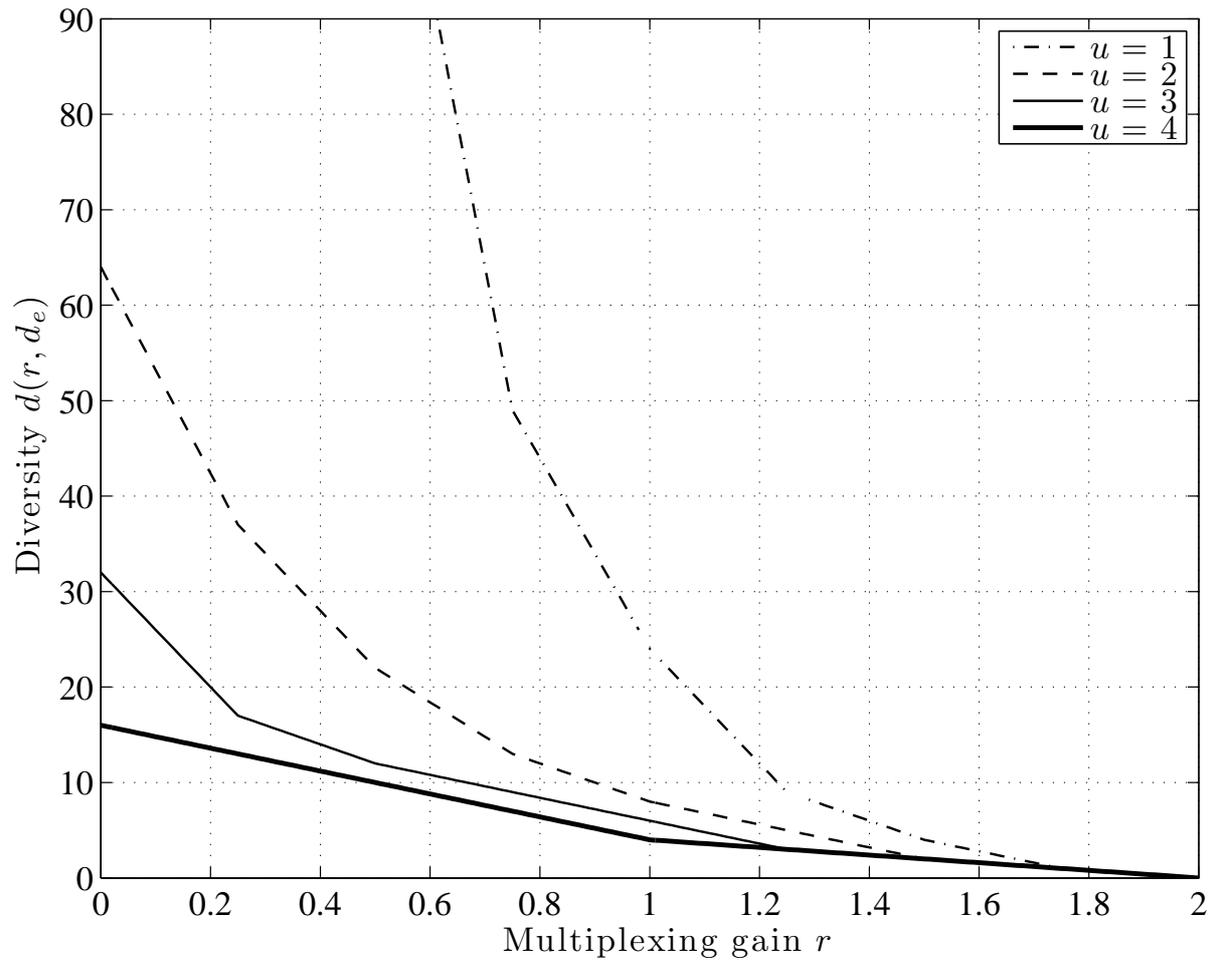}
  \caption{The achievable DMT of a 2-by-2 MIMO block-fading channel with $\noblock=4$ and perfect causal CSIT.}
  \label{fig:DMT_MIMO_del}
\end{figure}

\begin{figure}[htbp]
  \centering
  \includegraphics[width=1\columnwidth]{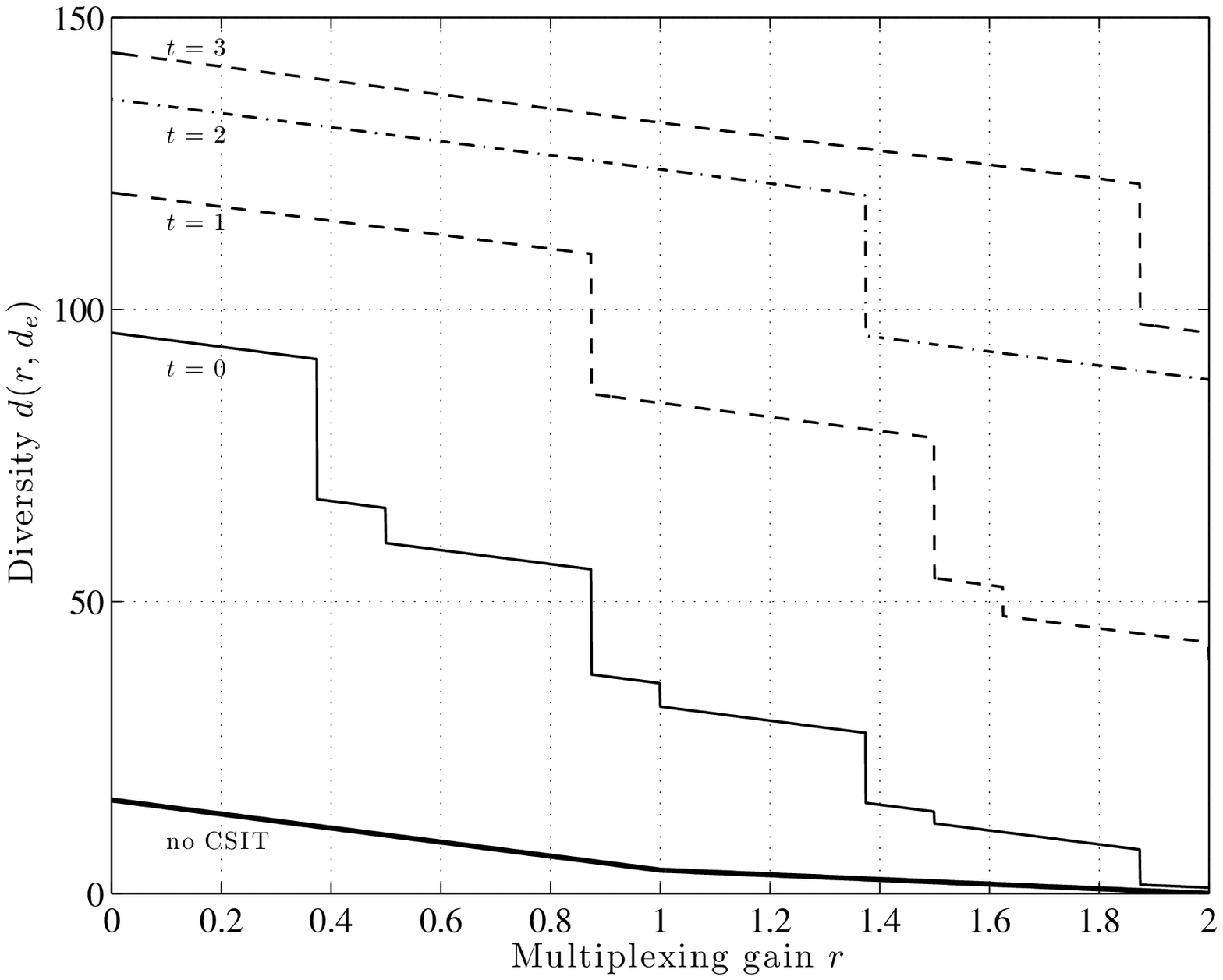}
  \caption{The achievable DMT of a 2-by-2 MIMO block-fading channel with $\noblock=4$ predictive mismatched CSIT $\de=0.5$}
  \label{fig:DMT_MIMO_pred_del}
\end{figure}

\begin{figure}[htbp]
  \centering
  \includegraphics[width=1\columnwidth]{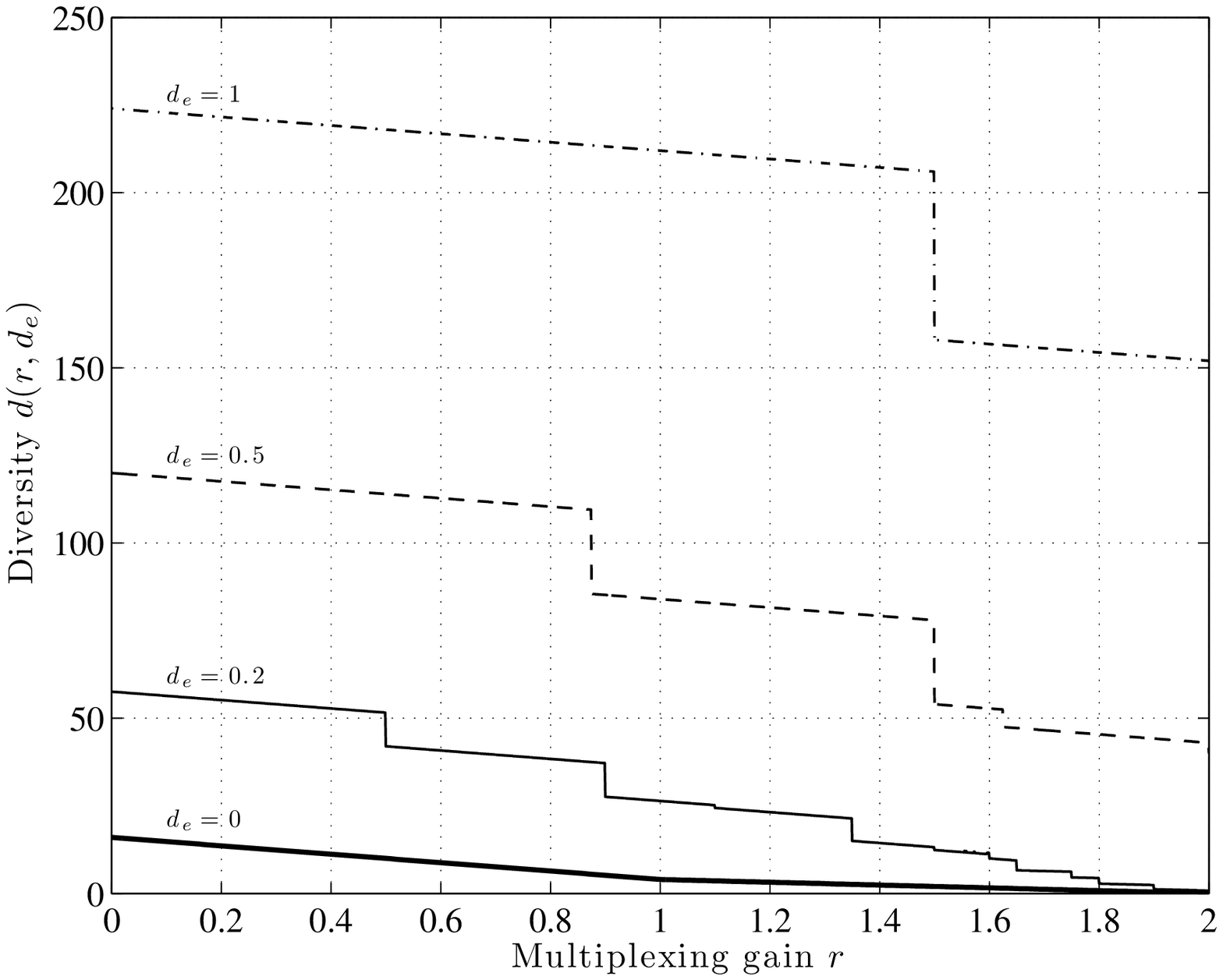}
  \caption{The achievable DMT of a 2-by-2 MIMO block-fading channel with $\noblock =4$, predictive mismatched CSIT $t=1$}
  \label{fig:DMT_MIMO_pred_va_de}
\end{figure}
%
%
%
\newpage
\begin{figure}[htbp]
  \centering
  \includegraphics[width =1 \columnwidth]{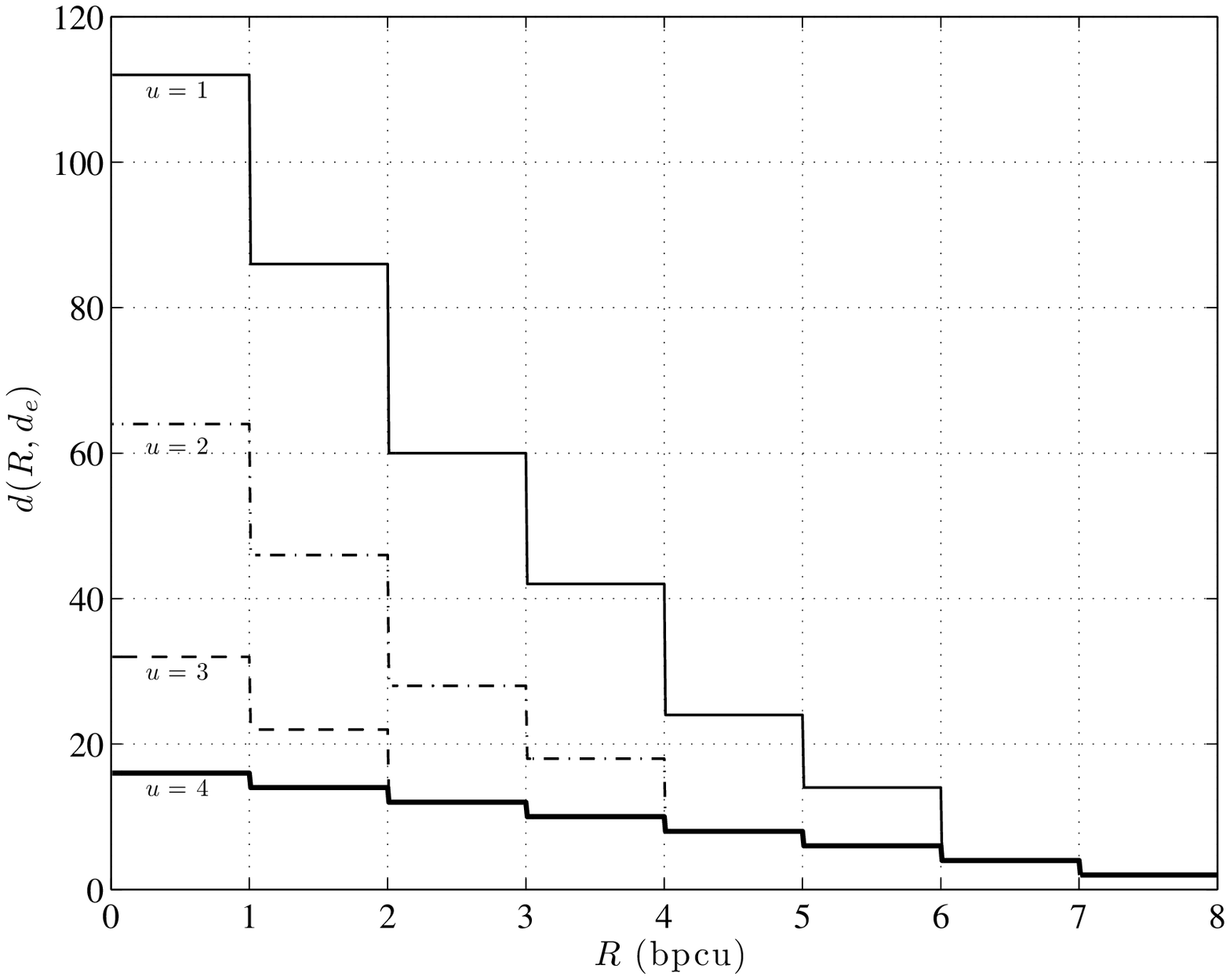}
  \caption{The optimal rate-diversity tradeoff for a 2-by-2 MIMO block-fading channel with $\noblock=4$ using 16-QAM input constellation, assuming  mismatch causal CSIT with exponent $\de= 1$.  }
  \label{fig:RDT_causal_mismatch_vsdel}
\end{figure}
\newpage
\begin{figure}[htbp]
  \centering
  \includegraphics[width=1\columnwidth]{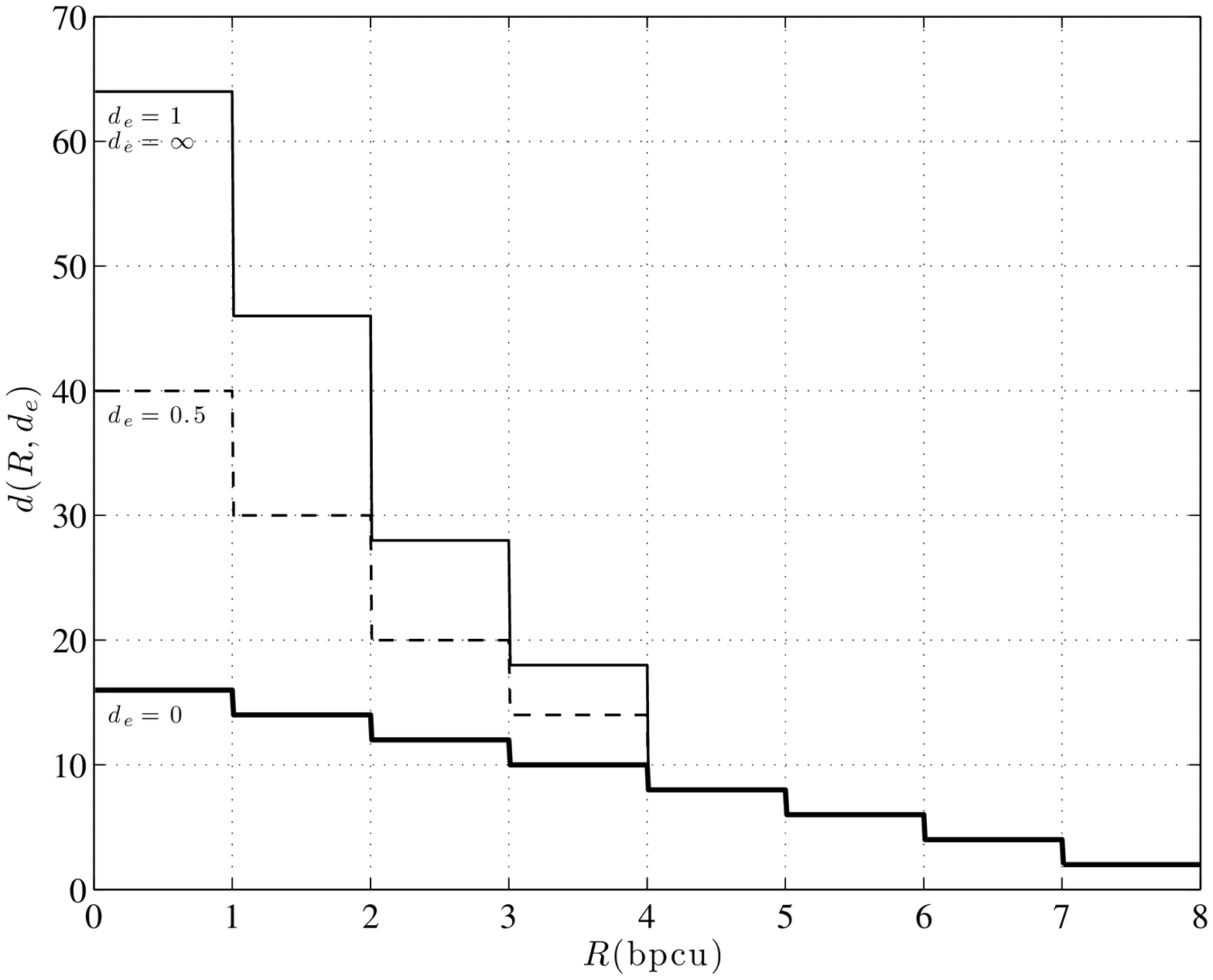}
  \caption{The optimal rate-diversity tradeoff for a 2-by-2 MIMO block-fading channel with $\noblock=4$ using 16-QAM input constellation, assuming causal CSIT with delay $\del=2$. }
  \label{fig:RDT_causal_mismatch_vsde}
\end{figure}
\newpage
\begin{figure}[htbp]
 \centering
  \includegraphics[width=1\columnwidth]{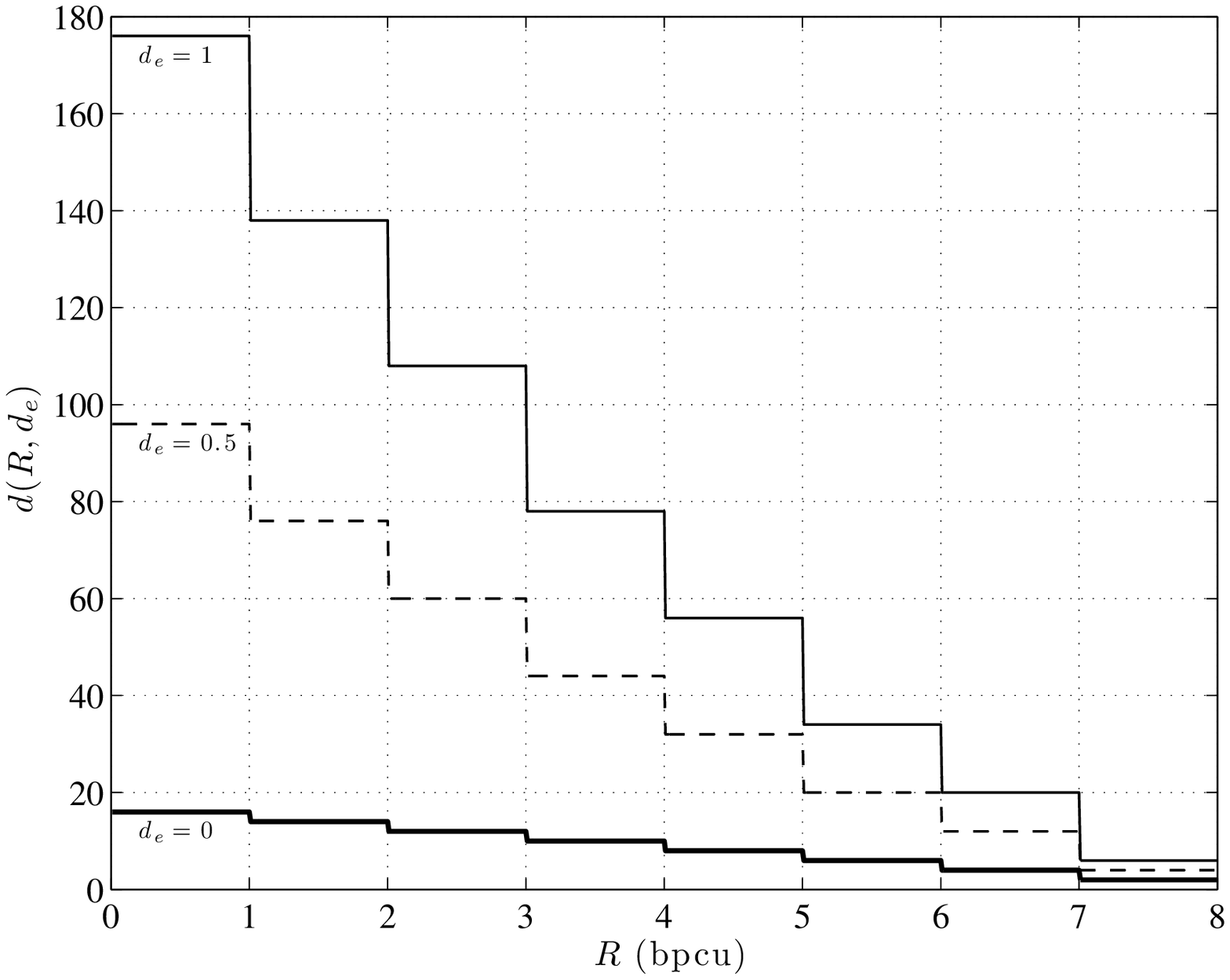}
  \caption{The optimal rate-diversity tradeoff for a 2-by-2 MIMO block-fading channel with $\noblock=4$ using 16-QAM input constellation, assumming mismatch predictive CSIT with $t=0$. }
  \label{fig:RDT_MIMO_predict_vsde}
\end{figure}
\newpage
\begin{figure}[htbp]
  \centering
  \includegraphics[width=1\columnwidth]{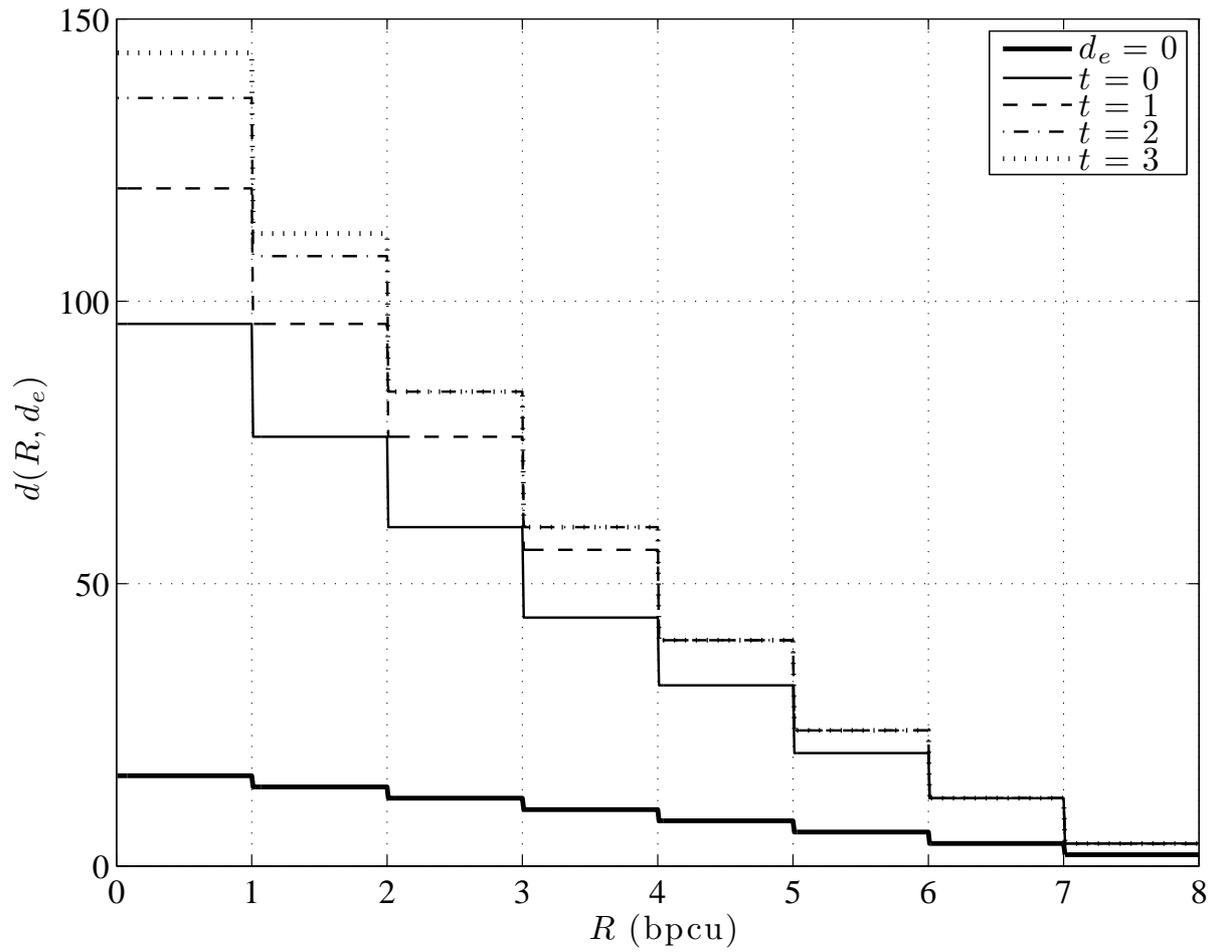} 
  \caption{The optimal rate-diversity tradeoff for a 2-by-2 MIMO block-fading channel with $\noblock=4$ using 16-QAM input constellation, with predictive CSIT of $t$ block and mismatched exponent $\de=0.5$. }
  \label{fig:RDT_MIMO_predict_vsdel}
\end{figure}
\end{document}